\documentclass[superscriptaddress,aps,pra,nofootinbib,onecolumn,notitlepage,10pt]{revtex4-1}
\pdfoutput=1
\usepackage{graphicx}
\usepackage{amsthm}
\usepackage{thmtools}
\usepackage{mathtools}
\usepackage{thm-restate}
\usepackage{amsmath}
\usepackage{amssymb}
\usepackage{comment}
\usepackage{placeins}
\usepackage[caption=false]{subfig}
\usepackage[colorlinks]{hyperref}
\usepackage{tikz}
\usetikzlibrary{calc,backgrounds}

\usepackage{pgffor}
\usepackage{url}
\usepackage{hypcap}
\usepackage{algorithm}
\usepackage{algorithmic}
\usepackage{soul} 
\usepackage{diagbox}

\newcommand{\eq}[1]{Eq.~\hyperref[eq:#1]{(\ref*{eq:#1})}}
\renewcommand{\sec}[1]{\hyperref[sec:#1]{Section~\ref*{sec:#1}}}
\newcommand{\app}[1]{\hyperref[app:#1]{Appendix~\ref*{app:#1}}}
\newcommand{\tab}[1]{\hyperref[tab:#1]{Table~\ref*{tab:#1}}}
\newcommand{\fig}[1]{\hyperref[fig:#1]{Figure~\ref*{fig:#1}}}
\newcommand{\figa}[2]{\hyperref[fig:#1]{Figure~\ref*{fig:#1}#2}}
\newcommand{\figx}[2]{\hyperref[fig:#1]{Figure~\ref*{fig:#1}(#2)}}
\newcommand{\thm}[1]{\hyperref[thm:#1]{Theorem~\ref*{thm:#1}}}
\newcommand{\lem}[1]{\hyperref[lem:#1]{Lemma~\ref*{lem:#1}}}
\newcommand{\cor}[1]{\hyperref[cor:#1]{Corollary~\ref*{cor:#1}}}
\newcommand{\defn}[1]{\hyperref[def:#1]{Definition~\ref*{def:#1}}}
\newcommand{\alg}[1]{\hyperref[alg:#1]{Algorithm~\ref*{alg:#1}}}

\def\bra#1{\mathinner{\langle{#1}|}}
\def\ket#1{\mathinner{|{#1}\rangle}}
\newcommand{\braket}[2]{\langle #1|#2\rangle}

\newcommand{\be}{\begin{equation}}
\newcommand{\ee}{\end{equation}}
\newcommand{\ba}{\begin{eqnarray}}
\newcommand{\ea}{\end{eqnarray}}

\usepackage{color}
\usepackage{xcolor}

\usepackage{colortbl}
\makeatletter

    \def\CT@@do@color{%
      \global\let\CT@do@color\relax
            \@tempdima\wd\z@
            \advance\@tempdima\@tempdimb
            \advance\@tempdima\@tempdimc
    \advance\@tempdimb\tabcolsep
    \advance\@tempdimc\tabcolsep
    \advance\@tempdima2\tabcolsep
            \kern-\@tempdimb
            \leaders\vrule
                    \hskip\@tempdima\@plus  1fill
            \kern-\@tempdimc
            \hskip-\wd\z@ \@plus -1fill }
    \makeatother

\newtheorem{theorem}{Theorem}
\newtheorem{lemma}[theorem]{Lemma}

\newenvironment{proofof}[1]{\begin{trivlist}\item[]{\flushleft\it
Proof of~#1.}}
{\qed\end{trivlist}}

\input{Qcircuit}

\begin{document}

\title{Phase estimation with randomized Hamiltonians}

\date{\today}
\author{Ian D.\ Kivlichan}
\affiliation{Department of Physics, Harvard University, Cambridge, MA 02138}
\affiliation{Department of Chemistry and Chemical Biology, Harvard University, Cambridge, MA 02138}
\author{Christopher E.\ Granade}
\affiliation{Microsoft Research, Redmond, WA 98052}
\author{Nathan Wiebe}
\affiliation{Microsoft Research, Redmond, WA 98052}

\begin{abstract}
Iterative phase estimation has long been used in quantum computing to estimate Hamiltonian eigenvalues. This is done by applying many repetitions of the same fundamental simulation circuit to an initial state, and using statistical inference to glean estimates of the eigenvalues from the resulting data.  
Here, we show a generalization of this framework where each of the steps in the simulation uses a different Hamiltonian.  
This allows the precision of the Hamiltonian to be changed as the phase estimation precision increases. Additionally, through the use of importance sampling, we can exploit knowledge about the ground state to decide how frequently each Hamiltonian term should appear in the evolution, and minimize the variance of our estimate.  
We rigorously show, if the Hamiltonian is gapped and the sample variance in the ground state expectation values of the Hamiltonian terms sufficiently small, that this process has a negligible impact on the resultant estimate and the success probability for phase estimation.
We demonstrate this process numerically for two chemical Hamiltonians, and observe substantial reductions in the number of terms in the Hamiltonian; in one case, we even observe a reduction in the number of qubits needed for the simulation.
Our results are agnostic to the particular simulation algorithm, and we expect these methods to be applicable to a range of approaches.
\end{abstract}
\maketitle

\section{Introduction}
Not all Hamiltonian terms are created equally in quantum simulation.  Hamiltonians that naturally arise from chemistry \cite{aspuru2005simulated,mcardle2018quantum,cao2018quantum}, materials \cite{babbush2018low} and other applications \cite{jordan2012quantum} are often composed of terms that are negligibly small.  These terms are  culled from the Hamiltonian well before it reaches the simulator.  Other terms that are formally present in the Hamiltonian are removed, not because of their norm, but rather because they are not expected to impact quantities of interest.  For example, in quantum chemistry, one usually selects an active space of orbitals and excludes  orbitals outside the active space \cite{helgaker2014molecular}.  This causes many large terms to be omitted from the Hamiltonian.

This process often involves systematically removing terms from the Hamiltonian and simulating the dynamics.  The idea behind such a scheme is to remove terms in the Hamiltonian until the maximum shift allowed in the eigenvalues is comparable to the level of precision needed.  For the case of chemistry, chemical accuracy sets a natural accuracy threshold for such simulations \cite{pople1999nobel}, but in general this precision requirement need not be viewed as a constant \cite{kivlichan2019improved}.

The principal insight of this work is that in iterative phase estimation the number of terms taken in the Hamiltonian should ideally not be held constant.  The reason why is that the high-order bits are mostly irrelevant when one is trying to learn, for example, a given bit of a binary expansion of the eigenphase.  A much lower accuracy simulation can be tolerated than it can when learning a high-order bit.  It then makes sense to adapt the number of terms in the Hamiltonian as iterative phase estimation proceeds through the bits of the phase estimation.  Our work proposes a systematic method for removing terms, and provides formal proofs that such processes need not dramatically affect the results of phase estimation nor its success probability.

The core idea behind our procedure is that we use a form of importance sampling to estimate, a priori, which terms in the Hamiltonian are significant, and from this generate randomized Hamiltonians which approximate the true one.  These randomized Hamiltonians are used within a simulation circuit to prepare approximate ground states.  We then show, using analysis reminiscent of that behind the Zeno effect or the quantum adiabatic theorem, that the errors in the eigenstate prepared at each round of phase estimation need not have a substantial impact on the posterior mean of the eigenphase estimated for the true Hamiltonian. 
This shows, under appropriate assumptions on the eigenvalue gaps, that this process can be used to reduce the time complexity of simulation; under some circumstances, even the space complexity may be reduced by identifying qubits that are not needed for the level of precision required of the simulation.

We proceed by first reviewing iterative phase estimation and Bayesian inference, which we use to quantify the maximal error in the inference of the phase.  
In the appendices, we examine the effect of using a stochastic Hamiltonian on the eigenphases yielded by phase estimation in the simple case where a fixed, but random, Hamiltonian is used at each step of iterative phase estimation.  
Using this result, we generalize to the more complicated case where each repetition of $e^{-iHt}$ in the iterative phase estimation circuit is implemented with a different random Hamiltonian.  
We end the theoretical analysis by showing how randomly sampling the Hamiltonians according to physically motivated importance functions can help to minimize the variance in the estimated phase. We hope this randomization may be applicable in other similar schemes which reduce the cost by reducing the number of simulated terms $e^{-iH_j \Delta t}$ rather than the number of terms $H_j$ in the Hamiltonian itself \cite{childs2018faster,campbell2018random}.
We further show in the appendices that the success probability is not degraded substantially if the eigenvalue gaps of the original Hamiltonian are sufficiently large. 
We end the paper by showing numerical examples of this sampling procedure, and from that conclude that our sampling process for the Hamiltonian can have a substantial impact on the number of terms in the Hamiltonian, and even in some cases the number of qubits used in the simulation.

\section{Iterative Phase Estimation}
The idea behind iterative phase estimation is simple.  We aim to build a quantum circuit that acts as an interferometer wherein the unitary we wish to probe is applied in one of the two branches but not the other.  When the quantum state is allowed to interfere with itself at the end of the protocol, the interference pattern reveals the eigenphase.  This process allows the eigenvalues of $U$ to be estimated within the standard quantum limit, i.e.\ $\Theta(1/\epsilon^2)$ applications of $U$ are needed to estimate the phase within error $\epsilon$.  If the quantum state is allowed to pass repeatedly through the interferometer circuit, or entangled inputs are used, then this scaling can be reduced to the Heisenberg limit $\Theta(1/\epsilon)$ \cite{giovannetti2004quantum,higgins2007entanglement,berry2009how}.  Such a circuit is shown in \fig{PEcircuit}.
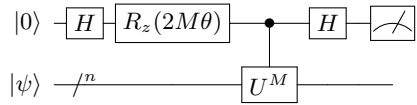
\begin{figure}[t]
\[\Qcircuit @R=1em @C=0.5em {
\lstick{\ket0} & \gate{H} & \gate{R_z(2M\theta)} & \ctrl{1} & \gate{H} & \qw & \meter \\
\lstick{\ket\psi} & {/^n} \qw & \qw & \gate{U^M} & \qw & \qw & \qw 
}\]
\caption{Quantum circuit for performing iterative phase estimation. $M$ is the number of repetitions (not necessarily an integer) of the controlled unitary $U$, and $\theta$ is a phase offset between the ancilla $\ket0$ and $\ket1$ states. $R_z(\varphi) = \exp(-i\varphi Z/2)$ for the Pauli operator $Z$, and $H$ is the Hadamard gate. \label{fig:PEcircuit}
}
\end{figure}

The phase estimation circuit is easy to analyze in the case where $U\ket{\psi} = e^{i\phi}\ket{\psi}$.  If $U$ is repeated $M$ times, and $\theta$ is a phase offset, then the likelihood of a given measurement outcome $o \in \{0, 1\}$ for the circuit in \fig{PEcircuit} is
\begin{equation}
\begin{aligned}
\Pr(o | \phi; M, \theta) = \frac{1 + (-1)^o \cos(M(\phi - \theta))}{2}.\label{eq:likelihood}
\end{aligned}
\end{equation}

There are many free parameters that can be used when designing iterative phase estimation experiments.  
In particular, the rules for generating $M$ and $\theta$ for each experiment vary radically along with the methods used to process the data from these experiments.  
Approaches such as Kitaev's phase estimation algorithm \cite{kitaev1995quantum}, robust phase estimation \cite{kimmel2015robust}, information theory phase estimation \cite{svore2014faster}, or any number of approximate Bayesian methods \cite{wiebe2015efficient,paesani2017experimental}, provide good heuristics for picking these parameters. 
In this work we do not wish to specify to any of these methods for choosing experiments, nor do we wish to focus on the specific data processing methods used.  
Nevertheless, we rely on Bayesian methods to discuss the impact that randomizing the Hamiltonian can have on an estimate of the eigenphase.

Bayes' theorem can be interpreted as giving the correct way to update beliefs about some fact given a set of experimental evidence and prior beliefs.  The initial beliefs of the experimentalist are encoded by a prior distribution $\Pr(\phi)$.  In many cases it is enough to set $\Pr(\phi)$ to be a uniform distribution on $[0,2\pi)$ to represent minimal assumptions about the eigenphase.  However, in quantum simulation broader priors can be chosen if each step in phase estimation uses $U_j=e^{-iHt_j}$ and obeys $U_j\ket{\psi}=e^{-iE_0t_j} \ket{\psi}$ for different $t_j$, since such experiments can learn $E_0$ as opposed to experiments with a fixed $t$ which yield $\phi = E_0 t\text{ mod } 2\pi$.

Bayes' theorem then gives the posterior distribution $\Pr(\phi|o;\phi,M)$ as
\begin{equation}
\Pr(\phi|o;M, \theta) = \frac{\Pr(o|\phi;M,\theta) \Pr(\phi)}{\int \Pr(o|\phi;M,\theta) \Pr(\phi) \,\mathrm{d}\phi}.
\end{equation}
Given a complete data set $(\vec o, \vec M, \vec\theta)$ rather than a single datum, we have that
\begin{equation}
\Pr(\phi | \vec{o};\vec{M},\vec{\theta}) = \frac{\prod_j\Pr(o_j|\phi;M_j,\theta_j) \Pr(\phi)}{\int \prod_j\Pr(o_j|\phi;M_j,\theta_j) \Pr(\phi) \,\mathrm{d}\phi}.
\end{equation}
This probability distribution encodes the experimentalist's entire state of knowledge about $\phi$ given that the data is processed optimally.

It is not customary to return the posterior distribution (or an approximation thereof) as output from a phase estimation protocol.  Instead, a point estimate for $\phi$ is given.  The most frequently used estimate is the maximum a posteriori (MAP) estimate, the $\phi$ with the maximum probability.  While this quantity has a nice operational interpretation, it suffers from several shortcomings for our purposes.  The main drawback here is that the MAP estimate is not robust, in the sense that if two different values of $\phi$ have comparable likelihoods, then small errors in the likelihood can lead to radical shifts in the MAP estimate.  
The posterior mean, $\int \Pr(\phi|\vec{o};\vec{M},\vec{\theta})\phi \,\mathrm{d}\phi$, is a better estimate for this purpose, as it minimizes the mean squared error in any unbiased estimate of $\phi$.  
It also has the property that it is robust to small perturbations in the likelihood---we will make use of this feature we to estimate the impact of our randomization procedure on the results of a phase estimation experiment.

\section{Results}
Our main result is a proof that the Hamiltonian used in iterative phase estimation can be randomized in between steps of iterative phase estimation.
This allows the user to employ heuristics to adaptively change the Hamiltonian as the precision increases, while also potentially reducing the number of terms used within a step of time evolution.
The idea for our procedure is that, rather than applying the controlled unitary $U$ generated by the true Hamiltonian $M$ times, we instead randomly generate a sequence of $M$ different Hamiltonians $\{ H_k \}_{k=1}^M$, and apply the controlled unitary operators defined by this sequence of randomized Hamiltonians in the phase estimation circuit.
The key quantity that we need to bound is the difference between the posterior mean that occurs due to randomizing the terms in the Hamiltonians.  
Note that we have not at this point specified a particular randomization scheme: our bound does not depend on it, except as the particular randomization scheme determines the minimum ground state gap of any of the generated Hamiltonian in the sequence and the maximum difference between any two consecutive Hamiltonians. 
We show this below, before elaborating on the randomization approach we use.
\begin{restatable}{theorem}{main}
Consider a sequence of Hamiltonians $\{ H_k \}_{k=1}^M$, $M>1$.
Let $\gamma$ be the minimum gap between the ground and first excited energies of any of the Hamiltonians, $\gamma = \min_k (E_1^k - E_0^k)$. Similarly, let $\lambda = \max_k \|H_k - H_{k-1}\|$ be the maximum difference between any two Hamiltonians in the sequence. The maximum error in the estimated eigenphases of the unitary found by the products of these $M$ Hamiltonians is at most 
$$
|\phi_{est} - \phi_{true}|\le \frac{2M\lambda^2}{(\gamma-2\lambda)^2},
$$
with a probability of failure of at most $\epsilon$ provided that
$$
\frac{ \lambda }{ \gamma } < \sqrt{1 - \exp\left( \frac{\log(1-\epsilon)}{ M-1 } \right)}.
$$
\label{thm:main}
\end{restatable}
The proof is contained in \app{phaseshift}. From this perspective, our randomization process can be seen as a generalized form of the term-elimination processes employed in previous quantum chemistry simulation work.
The fundamental idea behind our approach to eliminating Hamiltonian terms is importance sampling.  This approach has already seen great use in coalescing  \cite{wecker2014gate,poulin2015trotter}, but we use it slightly differently here.  The idea behind importance sampling is to reduce the variance in the mean of a quantity by reweighting the sum.  Specifically, we can write the mean of $N$ numbers $F(j)$ as
\begin{equation}
\frac{1}{N} \sum_j F(j) = \sum_j f(j)\frac{F(j)}{Nf(j)},
\end{equation}
where $f(j)$ is the importance assigned to a given term.   
This shows that we can view the initial unweighted average as the average of a reweighted quantity $x_j / (N f(j))$. 
While this does not have an impact on the mean of $x_j$, it can dramatically reduce the sample variance of the mean and thus is widely used in statistics to provide more accurate estimates of means.  
This motivates our approach to constructing sampled Hamiltonians: for a Hamiltonian given as a sum of terms $H = \sum_j H_j$, we assign to each term $H_j$ a (normalized) importance $f(j)$.  We then randomly construct the sequence of Hamiltonians $H_k$ as a sum of $N$ terms,
\begin{equation}
H_k = \frac{1}{N} \sum_{i=1}^N H_{\ell_i},
\end{equation}
where each $H_{\ell_i} = H_j / f(j)$ in the sum is sampled randomly independently from the different terms in the original Hamiltonian with probability given by the importance $f(j)$.
The mean of each $H_k$ is $H$. However, in (for example) Trotter-Suzuki-based simulation, the complexity depends on the number of terms \cite{poulin2015trotter}: the Hamiltonians generated in this way can have at most as many unique terms (different Pauli strings) as the original Hamiltonian but in many cases will have fewer. 
Returning to importance sampling, the optimal importance function to take is $f(j)\propto |x_j|$: in such cases, it is straightforward to see that the variance of the resulting distribution is in fact zero if the sign of the $x_j$ is constant.  A short calculation shows that this optimal variance is
\begin{equation}
\mathbb{V}_{f_{\rm opt}} = (\mathbb{E}(|F|))^2-(\mathbb{E}(F))^2. \label{eq:optVar}
\end{equation}

The optimal variance in \eq{optVar} is in fact zero if the sign of the numbers is constant.  While this may seem surprising, it becomes less mysterious by noting that in order to compute the optimal importance function one needs the ensemble mean one seeks to estimate.  This would defeat the purpose of importance sampling in most cases.  Thus, if we want to glean an advantage from importance sampling for Hamiltonian simulation, it is important to show that we can use it even with an inexact importance function that can be, for example, computed efficiently using a classical computer.
We show this robustness holds below.

\begin{lemma}\label{lem:robust}
Let $F:\mathbb{Z}_N \mapsto \mathbb{R}$  and let $\widetilde{F}:\mathbb{Z}_N\mapsto \mathbb{R}$ be a function such that for all $j$, $|\widetilde{F}(j)|-|F(j)|=\delta_j$ with $|\delta_j| \le |F(j)|/2$.  The variance from estimating $\mathbb{E}(F)$ using an importance function $f(j) = |\widetilde{F}(j)| / \sum_k |\widetilde{F}(k)|$ is
$$
\mathbb{V}_f (F) \le  \frac{4}{N^2}\left(\sum_k|\delta_k| \right)\left(\sum_j |F(j)|\right)+\mathbb{V}_{f_{\rm opt}}(F).
$$
\end{lemma}
We include the proof of this lemma in \app{importancesampling}. This bound is tight in the sense that as $\max_k |\delta_k|\rightarrow 0$ the upper bound on the variance converges to $\left(\mathbb{E}(|F|)\right)^2- \left(\mathbb{E}(F)\right)^2$ which is the optimal attainable variance.  
The key point this shows is that if an inexact importance function is used, then the variance varies smoothly with the error in the function.

In applications such as quantum chemistry simulation, our goal is to minimize the variance in~\lem{robust}.  
Then, this minimum variance can in principle be attained by choosing $f(j) \propto |\bra{\psi} H_j \ket{\psi}|$ for the term $H_j$ in the original Hamiltonian $H = \sum_j H_j$, where $\ket{\psi}$ is the eigenstate of interest.  
However, the task of computing this is at least as hard as solving the original eigenvalue estimation problem.  
The natural approach is to take inspiration from~\lem{robust}, and instead approximate the ideal importance function by taking $f(j) \propto |\!\bra{\widetilde{\psi}} H_j \ket{\widetilde{\psi}}\!|$, where $\ket{\widetilde{\psi}}$ is an efficiently computable ansatz state from e.g.\ truncated configuration interaction methods \cite{helgaker2014molecular}.  
In practice, however, the importance of a given term may not be entirely predicted by the ansatz prediction.  In this case, we can instead use a hedging strategy where for some $\rho \in [0,1]$, $f(j) \propto (1-\rho) \langle H_j \rangle + \rho \| H_j \|$.  
(The $\rho=1$ case is more comparable to recent work by Campbell \cite{campbell2018random}, though our approach reduces the cost of simulation by changing the Hamiltonian itself rather than randomizing the order in which terms are simulated as in that work; we hope that the importance sampling ideas discussed here may find application within or alongside that method.)  
This strategy allows us to smoothly interpolate between importance dictated by the magnitude of the Hamiltonian terms as well as the expectation value in the surrogate for the ground state.

\section{Numerical results}
Our work has shown that it is possible to use iterative phase estimation using a randomized Hamiltonian, but we have not discussed how effective this is in practice.  
We consider two examples of diatomic molecules, dilithium (\fig{dilithium}) and hydrogen chloride (\fig{hcl}).  In both cases, the molecules are prepared in a minimal STO-3G basis, and we use states found by variationally minimizing the ground state energy over all states within two excitations from the Hartree-Fock state (configuration interaction singles and doubles, or CISD) \cite{helgaker2014molecular}.  
We then randomly sample varying numbers of Hamiltonians terms for varying values of $\rho$, with the expectation value in the importance function $f(j) \propto (1-\rho) \langle H_j\rangle + \rho \|H_j \|$ is taken with respect to the CISD state. 
We examine several quantities of interest, including the average ground state energy, the variance in the ground state energies, and the average number of terms in the Hamiltonian.  Interestingly, we also look at the number of qubits present in the model.  
This can vary because some randomly sampled Hamiltonians will actually choose terms in the Hamiltonian that do not couple with the remainder of the system.  
In these cases, the number of qubits required to represent the state can in fact be lower than the total number that would be ordinarily required.

\begin{figure}[ht]
\centering
\subfloat[][]{
\includegraphics[width=0.36\textwidth]{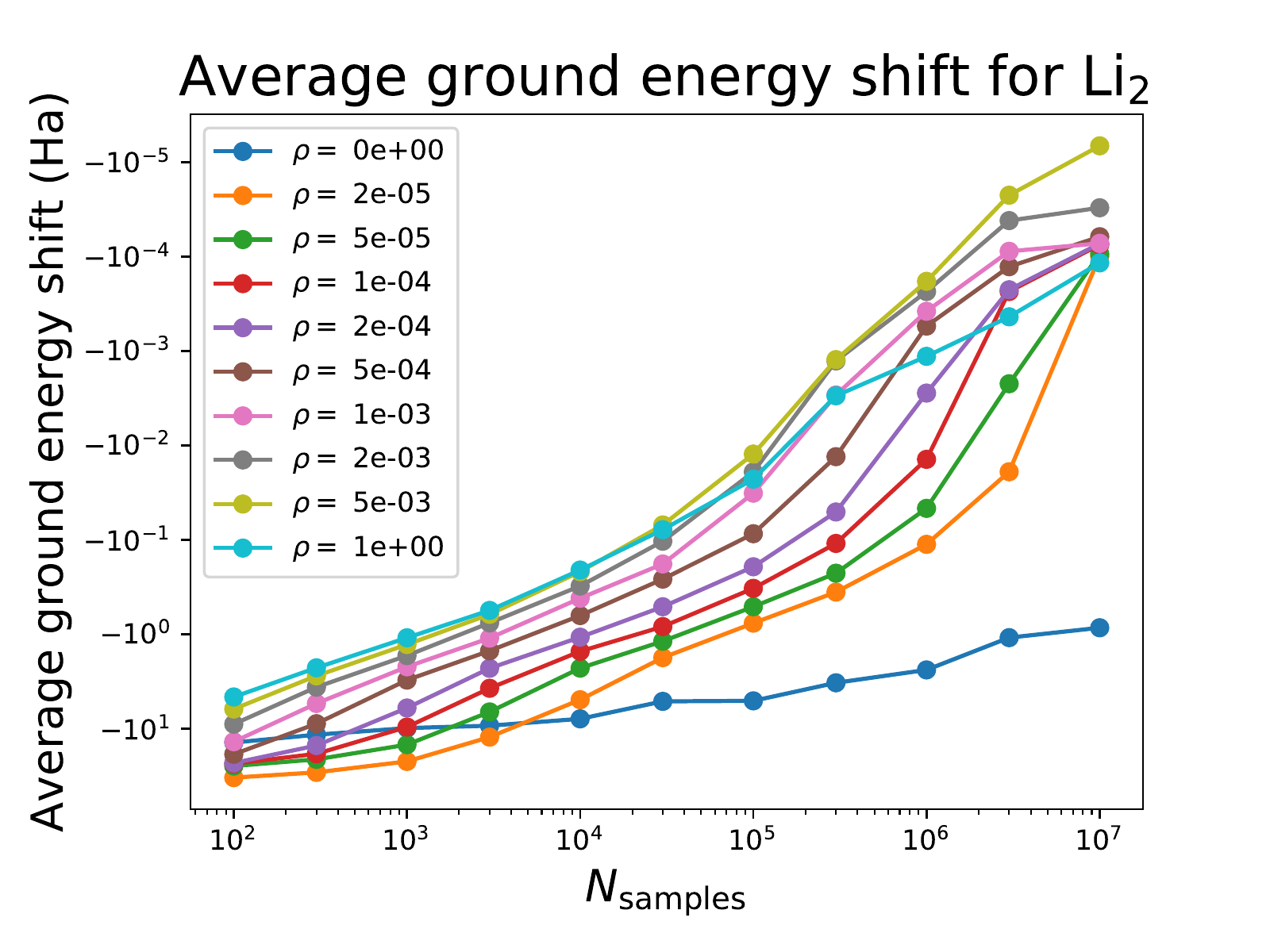}
\label{fig:Li2_avgshift}
}
\subfloat[][]{
\includegraphics[width=0.36\textwidth]{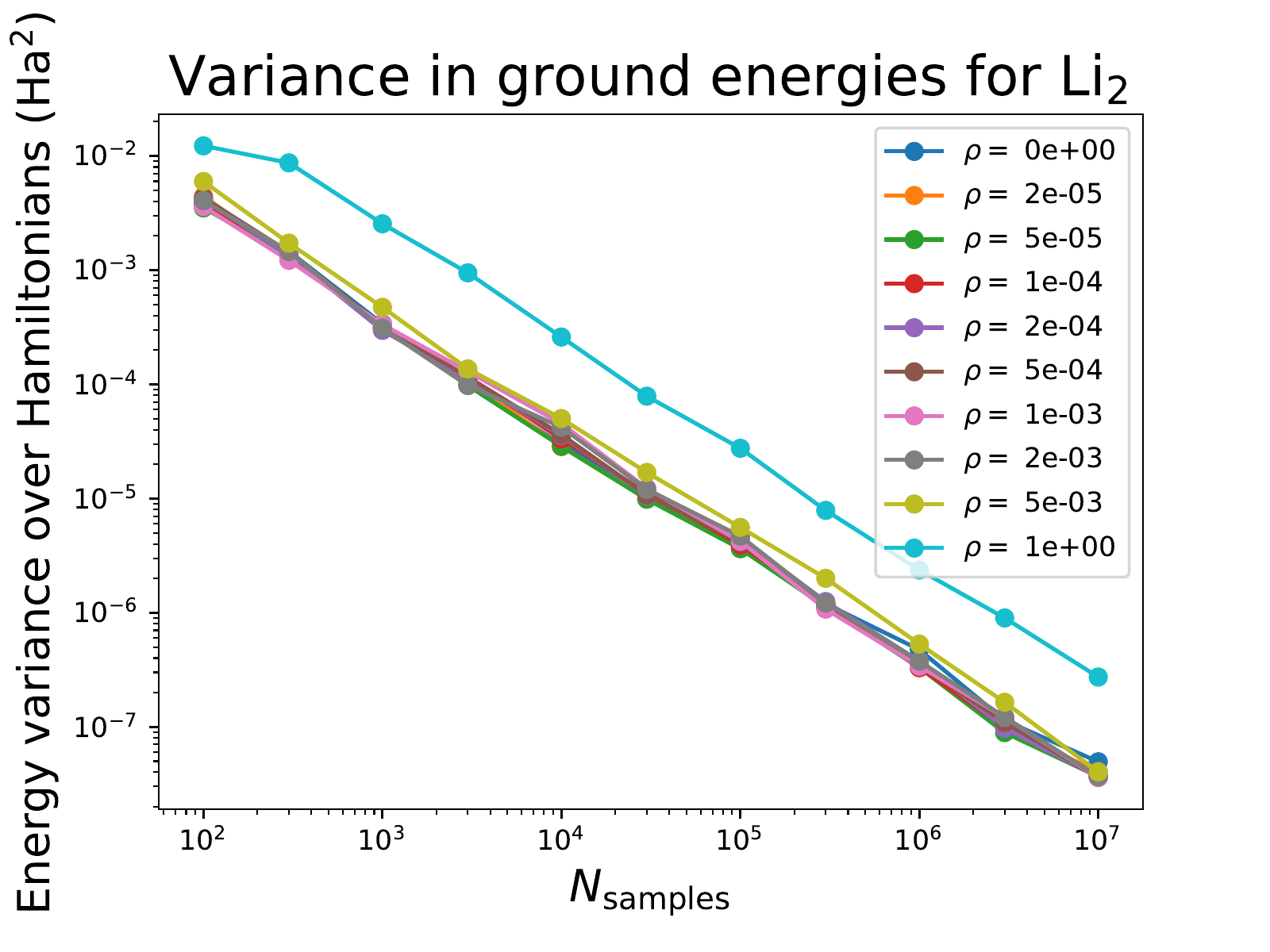}
\label{fig:Li2_var}
}
\\
\subfloat[][]{
\includegraphics[width=0.36\textwidth]{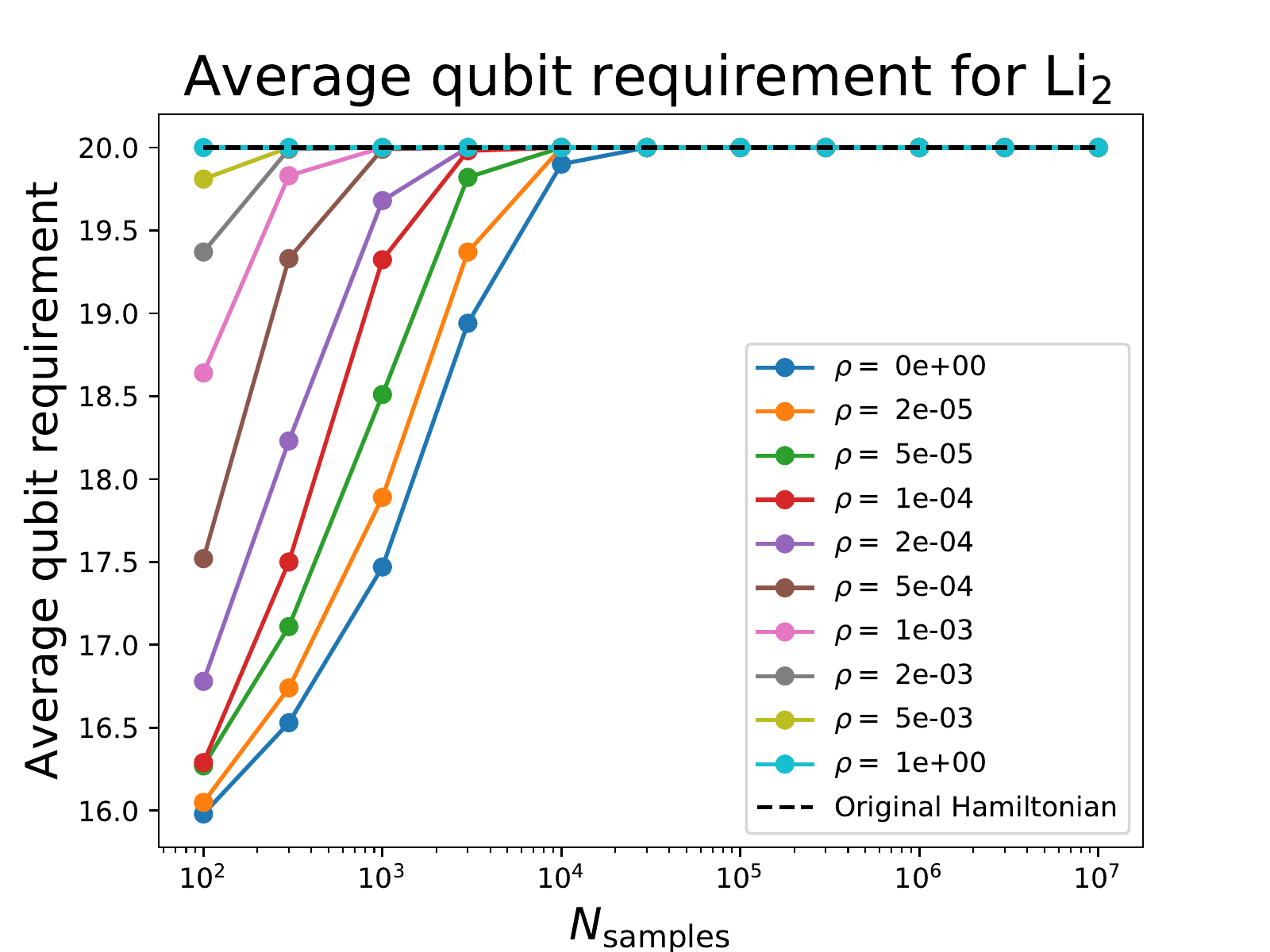}
\label{fig:Li2_avgqubits}
}
\subfloat[][]{
\includegraphics[width=0.36\textwidth]{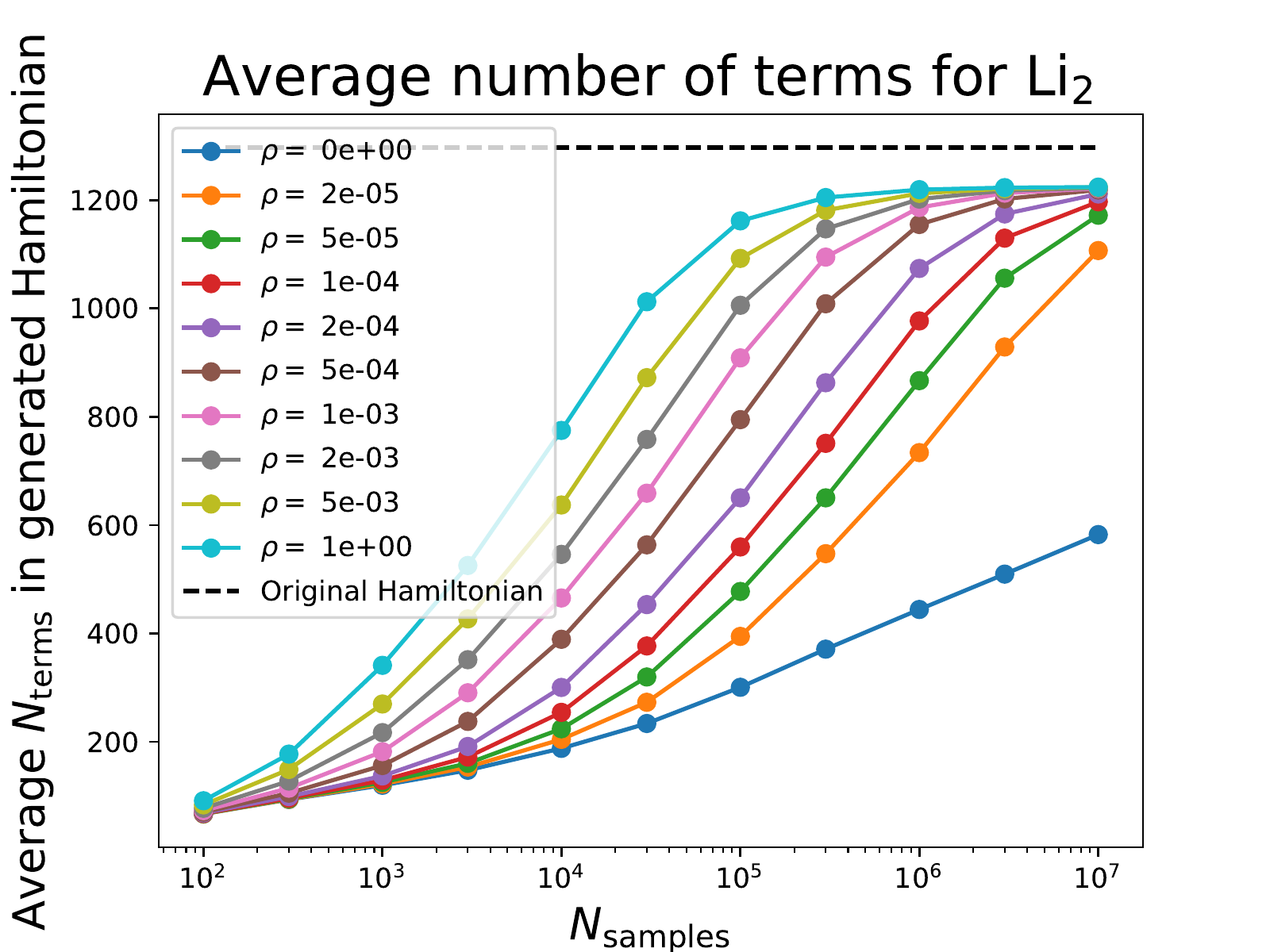}
\label{fig:Li2_avgterms}
}
\caption{\protect\subref{fig:Li2_avgshift} Average ground energy shift (compared to unsampled Hamiltonian), \protect\subref{fig:Li2_var} variance in ground energies over sampled Hamiltonians, \protect\subref{fig:Li2_avgqubits} average qubit requirement, and \protect\subref{fig:Li2_avgterms} average number of terms in sampled Hamiltonians for Li$_2$, as a function of number of samples taken to generate the Hamiltonian and the value of the parameter $\rho$. A term in the Hamiltonian $H_\alpha$ is sampled with probability $p_\alpha \propto (1-\rho) \langle H_\alpha \rangle + \rho \| H_\alpha \|$, where the expectation value is taken with the CISD state. 100 sampled Hamiltonians were randomly generated and numerically diagonalized for each data point. \label{fig:dilithium}}
\end{figure}

\begin{figure}[ht]
\centering
\subfloat[][]{
\includegraphics[width=0.36\textwidth]{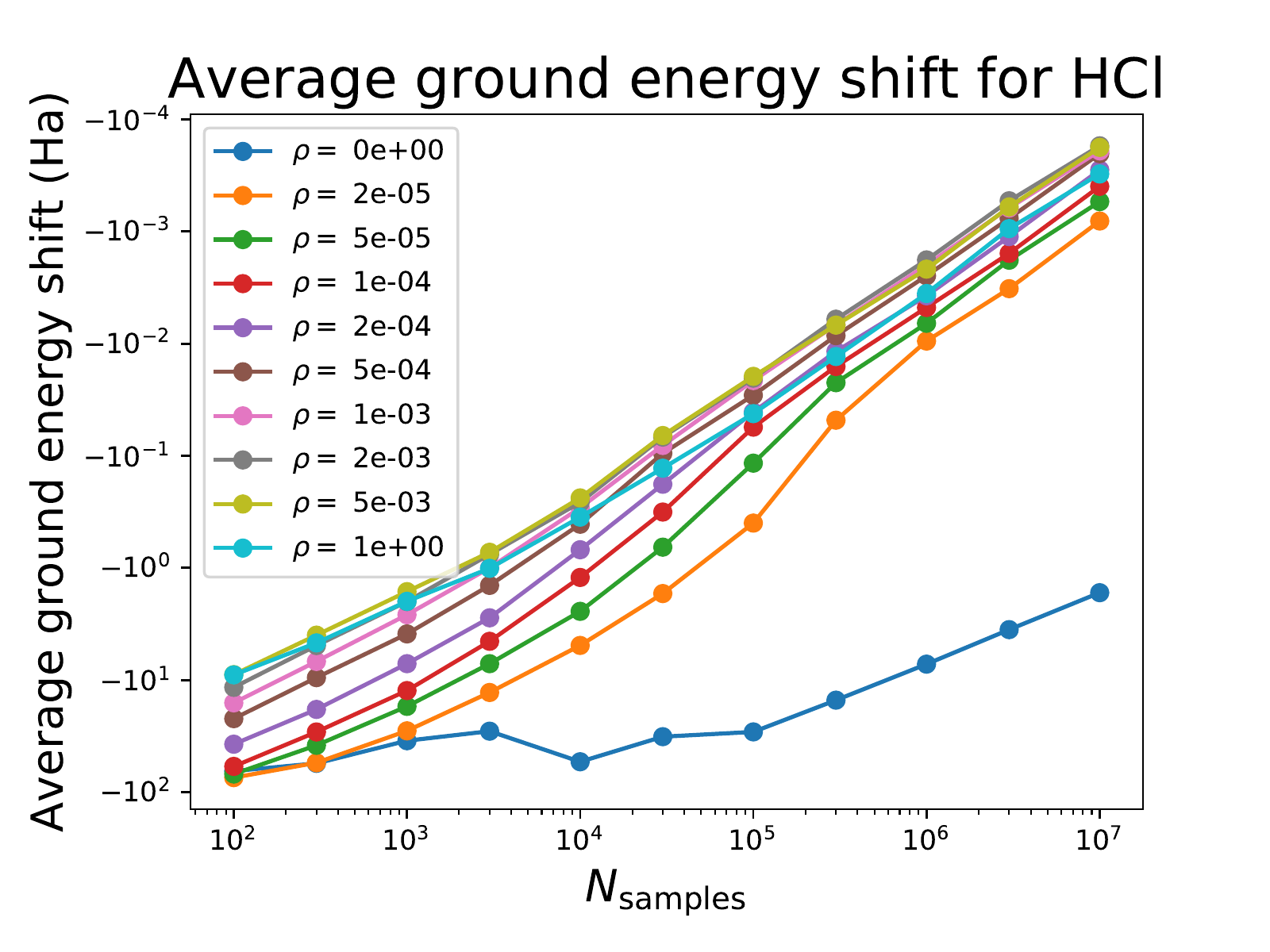}
\label{fig:HCl_avgshift}
}
\subfloat[][]{
\includegraphics[width=0.36\textwidth]{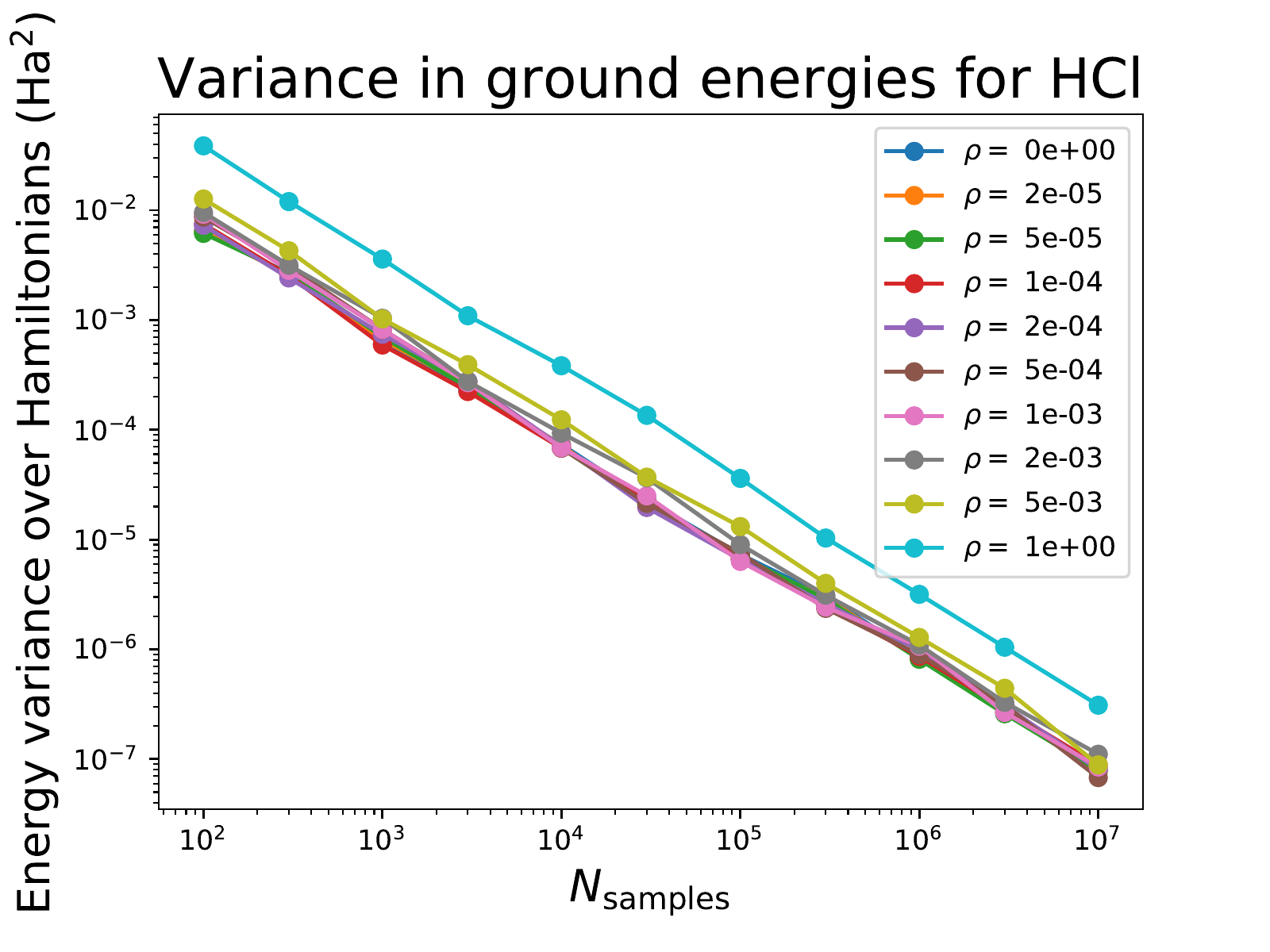}
\label{fig:HCl_var}
}
\\
\subfloat[][]{
\includegraphics[width=0.36\textwidth]{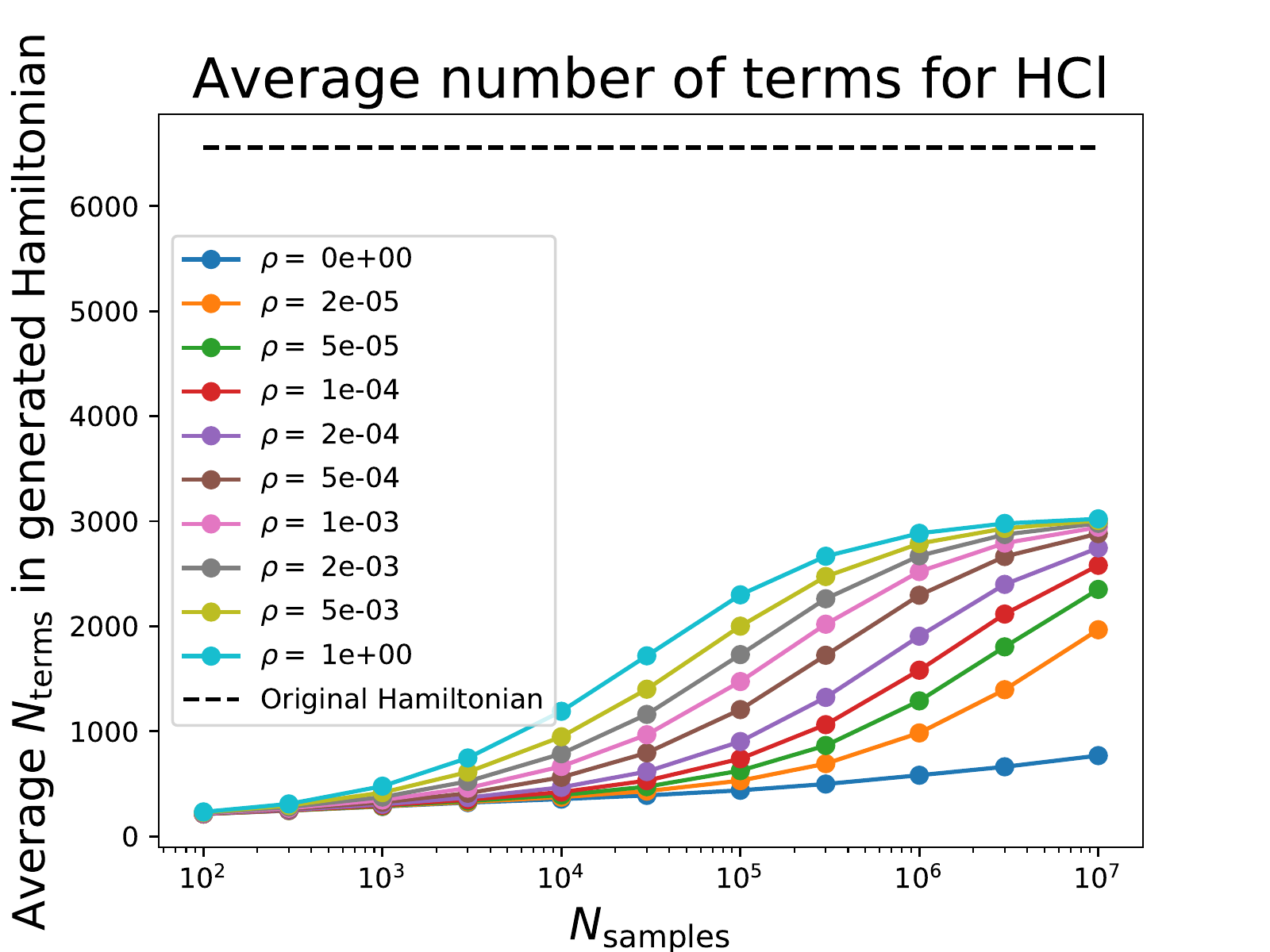}
\label{fig:HCl_avgterms}
}
\caption{\protect\subref{fig:HCl_avgshift} Average ground energy shift (compared to unsampled Hamiltonian), \protect\subref{fig:HCl_var} variance in ground energies over sampled Hamiltonians, and \protect\subref{fig:HCl_avgterms} average number of terms in sampled Hamiltonians for HCl, as a function of number of samples taken to generate the Hamiltonian and the value of the parameter $\rho$. A term in the Hamiltonian $H_\alpha$ is sampled with probability $p_\alpha \propto (1-\rho) \langle H_\alpha \rangle + \rho \| H_\alpha \|$, where the expectation value is taken with the CISD state. 100  sampled Hamiltonians were randomly generated and numerically diagonalized for each data point. Here, no reduction in qubit requirement is possible even for very small numbers of sampled terms; we do not display the corresponding plot. \label{fig:hcl}}
\end{figure}

We see in \fig{dilithium} and \fig{hcl} that the estimates of the ground state energy vary significantly with the degree of hedging used.  We find that if $\rho=1$ then regardless of the number of samples used in constructing the Hamiltonians that we have a very large variance in the ground state energy, as expected since importance sampling has very little impact in that case.  Conversely, we find that if we take $\rho=0$ we maximally privilege the importance of Hamiltonian terms from the CISD state. This leads to very concise models but with shifts in ground state energies that are on the order of $1$ Hartree (Ha) for even $10^{7}$ randomly selected terms (some of which may be duplicates).  If we instead use a modest amount of hedging ($\rho =2\times 10^{-5}$) then we notice that the shift in the ground state energy is minimized assuming that a shift in energy of $10\%$ of chemical accuracy or $0.1$ mHa is acceptable for Hamiltonian truncation error.  For dilithium, this represents a $30\%$ reduction in the number of terms in the Hamiltonian; on the other hand, for hydrogen chloride, this reduces the number of terms in the Hamiltonian by a factor of $3$.  Since the cost of a Trotter-Suzuki simulation of chemistry scales super-linearly with the number of terms in the Hamiltonian this constitutes a substantial reduction in the complexity.

We also note that for the case of dilithium, the number of qubits needed to perform the simulation varied over the different runs.  By contrast hydrogen chloride showed no such behavior.  This difference arises from the fact that dilithium requires six electrons that reside in 20 spin orbitals.  Hydrogen chloride consists of eighteen electrons again in 20 spin orbitals.  As a result, nearly every spin orbital will be relevant in that which explains why the number of spin orbitals needed to express dilithium to a fixed degree of precision changes whereas it does not for HCl.  This illustrates that our randomization procedure can be used to help select an active space for a simulation on the fly as the precision needed in the Hamiltonian increases through a phase estimation procedure.

\section{Conclusion}
This work has shown that iterative phase estimation is more flexible than may previously have been thought, and that the number of terms in the Hamiltonian be randomized at each step of iterative phase estimation without substantially contributing to the underlying variance of an unbiased estimator of the eigenphase.  We further show numerically that by using such strategies for subsampling the Hamiltonian terms, that we can perform a simulation using fewer Hamiltonian terms than traditional approaches require.  
These reductions in the number of terms directly impact the complexity of Trotter-Suzuki-based simulation and indirectly impact qubitization and truncated Taylor series simulation methods \cite{berry2015simulating,low2016hamiltonian} because they also reduce the 1-norm of the vector of Hamiltonian terms.  
The simpler idea of using a sampled Hamiltonian with fewer terms (without using it in phase estimation) could also lead to improvements for variational algorithms \cite{peruzzo2014variational}; to our knowledge this direction has not been explored.

Now that we have shown that iterative phase estimation can be applied using a randomized Hamiltonian, a number of options arise for future work.  While we looked at using CISD states to compute the importance of terms in the Hamiltonian, other methods can also be used such as coupled cluster ansatzes \cite{helgaker2014molecular}.  
Additionally, machine learning algorithms could be used to find more sophisticated importance functionals than the hedged functions that we considered.  
More broadly, these methods could potentially be used with other approaches such as coalescing to further reduce the number of terms present in the Hamiltonian.  A systematic study of this may reveal a generalized method for reducing the cost of simulation that incorporates not just methods to reduce the number of Hamiltonian terms but more generic structural optimizations to the Hamiltonian as well.

\subsection*{Acknowledgments}
We thank Dave Wecker for several illuminating discussions about coalescing, Matthias Degroote for feedback on an earlier version of the manuscript, as well as the contributors to the open-source packages OpenFermion~\cite{openfermion2017} and Psi4~\cite{turney2012psi4} which were used for some of the numerics in this work. I.~D.~K.~acknowledges partial support from the National Sciences and Engineering Research Council of Canada.

\bibliographystyle{apsrev4-1}
\bibliography{science}

\begin{thebibliography}{25}%
\makeatletter
\providecommand \@ifxundefined [1]{%
 \@ifx{#1\undefined}
}%
\providecommand \@ifnum [1]{%
 \ifnum #1\expandafter \@firstoftwo
 \else \expandafter \@secondoftwo
 \fi
}%
\providecommand \@ifx [1]{%
 \ifx #1\expandafter \@firstoftwo
 \else \expandafter \@secondoftwo
 \fi
}%
\providecommand \natexlab [1]{#1}%
\providecommand \enquote  [1]{``#1''}%
\providecommand \bibnamefont  [1]{#1}%
\providecommand \bibfnamefont [1]{#1}%
\providecommand \citenamefont [1]{#1}%
\providecommand \href@noop [0]{\@secondoftwo}%
\providecommand \href [0]{\begingroup \@sanitize@url \@href}%
\providecommand \@href[1]{\@@startlink{#1}\@@href}%
\providecommand \@@href[1]{\endgroup#1\@@endlink}%
\providecommand \@sanitize@url [0]{\catcode `\\12\catcode `\$12\catcode
  `\&12\catcode `\#12\catcode `\^12\catcode `\_12\catcode `\%12\relax}%
\providecommand \@@startlink[1]{}%
\providecommand \@@endlink[0]{}%
\providecommand \url  [0]{\begingroup\@sanitize@url \@url }%
\providecommand \@url [1]{\endgroup\@href {#1}{\urlprefix }}%
\providecommand \urlprefix  [0]{URL }%
\providecommand \Eprint [0]{\href }%
\providecommand \doibase [0]{http://dx.doi.org/}%
\providecommand \selectlanguage [0]{\@gobble}%
\providecommand \bibinfo  [0]{\@secondoftwo}%
\providecommand \bibfield  [0]{\@secondoftwo}%
\providecommand \translation [1]{[#1]}%
\providecommand \BibitemOpen [0]{}%
\providecommand \bibitemStop [0]{}%
\providecommand \bibitemNoStop [0]{.\EOS\space}%
\providecommand \EOS [0]{\spacefactor3000\relax}%
\providecommand \BibitemShut  [1]{\csname bibitem#1\endcsname}%
\let\auto@bib@innerbib\@empty
\bibitem [{\citenamefont {Aspuru-Guzik}\ \emph {et~al.}(2005)\citenamefont
  {Aspuru-Guzik}, \citenamefont {Dutoi}, \citenamefont {Love},\ and\
  \citenamefont {Head-Gordon}}]{aspuru2005simulated}%
  \BibitemOpen
  \bibfield  {author} {\bibinfo {author} {\bibfnamefont {A.}~\bibnamefont
  {Aspuru-Guzik}}, \bibinfo {author} {\bibfnamefont {A.~D.}\ \bibnamefont
  {Dutoi}}, \bibinfo {author} {\bibfnamefont {P.~J.}\ \bibnamefont {Love}}, \
  and\ \bibinfo {author} {\bibfnamefont {M.}~\bibnamefont {Head-Gordon}},\
  }\href {\doibase 10.1126/science.1113479} {\bibfield  {journal} {\bibinfo
  {journal} {Science}\ }\textbf {\bibinfo {volume} {309}},\ \bibinfo {pages}
  {1704} (\bibinfo {year} {2005})}\BibitemShut {NoStop}%
\bibitem [{\citenamefont {McArdle}\ \emph {et~al.}(2018)\citenamefont
  {McArdle}, \citenamefont {Endo}, \citenamefont {Aspuru-Guzik}, \citenamefont
  {Benjamin},\ and\ \citenamefont {Yuan}}]{mcardle2018quantum}%
  \BibitemOpen
  \bibfield  {author} {\bibinfo {author} {\bibfnamefont {S.}~\bibnamefont
  {McArdle}}, \bibinfo {author} {\bibfnamefont {S.}~\bibnamefont {Endo}},
  \bibinfo {author} {\bibfnamefont {A.}~\bibnamefont {Aspuru-Guzik}}, \bibinfo
  {author} {\bibfnamefont {S.}~\bibnamefont {Benjamin}}, \ and\ \bibinfo
  {author} {\bibfnamefont {X.}~\bibnamefont {Yuan}},\ }\href
  {http://arxiv.org/abs/1808.10402} {\bibfield  {journal} {\bibinfo  {journal}
  {arXiv preprint arXiv:1808.10402}\ } (\bibinfo {year} {2018})}\BibitemShut
  {NoStop}%
\bibitem [{\citenamefont {Cao}\ \emph {et~al.}(2018)\citenamefont {Cao},
  \citenamefont {Romero}, \citenamefont {Olson}, \citenamefont {Degroote},
  \citenamefont {Johnson}, \citenamefont {Kieferov{\'a}}, \citenamefont
  {Kivlichan}, \citenamefont {Menke}, \citenamefont {Peropadre}, \citenamefont
  {Sawaya} \emph {et~al.}}]{cao2018quantum}%
  \BibitemOpen
  \bibfield  {author} {\bibinfo {author} {\bibfnamefont {Y.}~\bibnamefont
  {Cao}}, \bibinfo {author} {\bibfnamefont {J.}~\bibnamefont {Romero}},
  \bibinfo {author} {\bibfnamefont {J.~P.}\ \bibnamefont {Olson}}, \bibinfo
  {author} {\bibfnamefont {M.}~\bibnamefont {Degroote}}, \bibinfo {author}
  {\bibfnamefont {P.~D.}\ \bibnamefont {Johnson}}, \bibinfo {author}
  {\bibfnamefont {M.}~\bibnamefont {Kieferov{\'a}}}, \bibinfo {author}
  {\bibfnamefont {I.~D.}\ \bibnamefont {Kivlichan}}, \bibinfo {author}
  {\bibfnamefont {T.}~\bibnamefont {Menke}}, \bibinfo {author} {\bibfnamefont
  {B.}~\bibnamefont {Peropadre}}, \bibinfo {author} {\bibfnamefont {N.~P.}\
  \bibnamefont {Sawaya}},  \emph {et~al.},\ }\href
  {http://arxiv.org/abs/1812.09976} {\bibfield  {journal} {\bibinfo  {journal}
  {arXiv preprint arXiv:1812.09976}\ } (\bibinfo {year} {2018})}\BibitemShut
  {NoStop}%
\bibitem [{\citenamefont {Babbush}\ \emph {et~al.}(2018)\citenamefont
  {Babbush}, \citenamefont {Wiebe}, \citenamefont {McClean}, \citenamefont
  {McClain}, \citenamefont {Neven},\ and\ \citenamefont
  {Chan}}]{babbush2018low}%
  \BibitemOpen
  \bibfield  {author} {\bibinfo {author} {\bibfnamefont {R.}~\bibnamefont
  {Babbush}}, \bibinfo {author} {\bibfnamefont {N.}~\bibnamefont {Wiebe}},
  \bibinfo {author} {\bibfnamefont {J.}~\bibnamefont {McClean}}, \bibinfo
  {author} {\bibfnamefont {J.}~\bibnamefont {McClain}}, \bibinfo {author}
  {\bibfnamefont {H.}~\bibnamefont {Neven}}, \ and\ \bibinfo {author}
  {\bibfnamefont {G.~K.-L.}\ \bibnamefont {Chan}},\ }\href {\doibase
  10.1103/PhysRevX.8.011044} {\bibfield  {journal} {\bibinfo  {journal} {Phys.
  Rev. X}\ }\textbf {\bibinfo {volume} {8}},\ \bibinfo {pages} {011044}
  (\bibinfo {year} {2018})}\BibitemShut {NoStop}%
\bibitem [{\citenamefont {Jordan}\ \emph {et~al.}(2012)\citenamefont {Jordan},
  \citenamefont {Lee},\ and\ \citenamefont {Preskill}}]{jordan2012quantum}%
  \BibitemOpen
  \bibfield  {author} {\bibinfo {author} {\bibfnamefont {S.~P.}\ \bibnamefont
  {Jordan}}, \bibinfo {author} {\bibfnamefont {K.~S.~M.}\ \bibnamefont {Lee}},
  \ and\ \bibinfo {author} {\bibfnamefont {J.}~\bibnamefont {Preskill}},\
  }\href {\doibase 10.1126/science.1217069} {\bibfield  {journal} {\bibinfo
  {journal} {Science}\ }\textbf {\bibinfo {volume} {336}},\ \bibinfo {pages}
  {1130} (\bibinfo {year} {2012})}\BibitemShut {NoStop}%
\bibitem [{\citenamefont {Helgaker}\ \emph {et~al.}(2014)\citenamefont
  {Helgaker}, \citenamefont {Jorgensen},\ and\ \citenamefont
  {Olsen}}]{helgaker2014molecular}%
  \BibitemOpen
  \bibfield  {author} {\bibinfo {author} {\bibfnamefont {T.}~\bibnamefont
  {Helgaker}}, \bibinfo {author} {\bibfnamefont {P.}~\bibnamefont {Jorgensen}},
  \ and\ \bibinfo {author} {\bibfnamefont {J.}~\bibnamefont {Olsen}},\
  }\href@noop {} {\emph {\bibinfo {title} {Molecular Electronic-Structure
  Theory}}}\ (\bibinfo  {publisher} {John Wiley \& Sons},\ \bibinfo {year}
  {2014})\BibitemShut {NoStop}%
\bibitem [{\citenamefont {Pople}(1999)}]{pople1999nobel}%
  \BibitemOpen
  \bibfield  {author} {\bibinfo {author} {\bibfnamefont {J.~A.}\ \bibnamefont
  {Pople}},\ }\href {\doibase 10.1103/RevModPhys.71.1267} {\bibfield  {journal}
  {\bibinfo  {journal} {Rev. Mod. Phys.}\ }\textbf {\bibinfo {volume} {71}},\
  \bibinfo {pages} {1267} (\bibinfo {year} {1999})}\BibitemShut {NoStop}%
\bibitem [{\citenamefont {Kivlichan}\ \emph {et~al.}(2019)\citenamefont
  {Kivlichan}, \citenamefont {Gidney}, \citenamefont {Berry}, \citenamefont
  {Wiebe}, \citenamefont {McClean}, \citenamefont {Sun}, \citenamefont {Jiang},
  \citenamefont {Rubin}, \citenamefont {Fowler}, \citenamefont {Aspuru-Guzik}
  \emph {et~al.}}]{kivlichan2019improved}%
  \BibitemOpen
  \bibfield  {author} {\bibinfo {author} {\bibfnamefont {I.~D.}\ \bibnamefont
  {Kivlichan}}, \bibinfo {author} {\bibfnamefont {C.}~\bibnamefont {Gidney}},
  \bibinfo {author} {\bibfnamefont {D.~W.}\ \bibnamefont {Berry}}, \bibinfo
  {author} {\bibfnamefont {N.}~\bibnamefont {Wiebe}}, \bibinfo {author}
  {\bibfnamefont {J.}~\bibnamefont {McClean}}, \bibinfo {author} {\bibfnamefont
  {W.}~\bibnamefont {Sun}}, \bibinfo {author} {\bibfnamefont {Z.}~\bibnamefont
  {Jiang}}, \bibinfo {author} {\bibfnamefont {N.}~\bibnamefont {Rubin}},
  \bibinfo {author} {\bibfnamefont {A.}~\bibnamefont {Fowler}}, \bibinfo
  {author} {\bibfnamefont {A.}~\bibnamefont {Aspuru-Guzik}},  \emph {et~al.},\
  }\href {https://arxiv.org/abs/1902.10673} {\bibfield  {journal} {\bibinfo
  {journal} {arXiv preprint arXiv:1902.10673}\ } (\bibinfo {year}
  {2019})}\BibitemShut {NoStop}%
\bibitem [{\citenamefont {Childs}\ \emph {et~al.}(2018)\citenamefont {Childs},
  \citenamefont {Ostrander},\ and\ \citenamefont {Su}}]{childs2018faster}%
  \BibitemOpen
  \bibfield  {author} {\bibinfo {author} {\bibfnamefont {A.~M.}\ \bibnamefont
  {Childs}}, \bibinfo {author} {\bibfnamefont {A.}~\bibnamefont {Ostrander}}, \
  and\ \bibinfo {author} {\bibfnamefont {Y.}~\bibnamefont {Su}},\ }\href
  {http://arxiv.org/abs/1805.08385} {\bibfield  {journal} {\bibinfo  {journal}
  {arXiv preprint arXiv:1805.08385}\ } (\bibinfo {year} {2018})}\BibitemShut
  {NoStop}%
\bibitem [{\citenamefont {Campbell}(2018)}]{campbell2018random}%
  \BibitemOpen
  \bibfield  {author} {\bibinfo {author} {\bibfnamefont {E.}~\bibnamefont
  {Campbell}},\ }\href {http://arxiv.org/abs/1811.08017} {\bibfield  {journal}
  {\bibinfo  {journal} {arXiv preprint arXiv:1811.08017}\ } (\bibinfo {year}
  {2018})}\BibitemShut {NoStop}%
\bibitem [{\citenamefont {Giovannetti}\ \emph {et~al.}(2004)\citenamefont
  {Giovannetti}, \citenamefont {Lloyd},\ and\ \citenamefont
  {Maccone}}]{giovannetti2004quantum}%
  \BibitemOpen
  \bibfield  {author} {\bibinfo {author} {\bibfnamefont {V.}~\bibnamefont
  {Giovannetti}}, \bibinfo {author} {\bibfnamefont {S.}~\bibnamefont {Lloyd}},
  \ and\ \bibinfo {author} {\bibfnamefont {L.}~\bibnamefont {Maccone}},\ }\href
  {\doibase 10.1126/science.1104149} {\bibfield  {journal} {\bibinfo  {journal}
  {Science}\ }\textbf {\bibinfo {volume} {306}},\ \bibinfo {pages} {1330}
  (\bibinfo {year} {2004})}\BibitemShut {NoStop}%
\bibitem [{\citenamefont {Higgins}\ \emph {et~al.}(2007)\citenamefont
  {Higgins}, \citenamefont {Berry}, \citenamefont {Bartlett}, \citenamefont
  {Wiseman},\ and\ \citenamefont {Pryde}}]{higgins2007entanglement}%
  \BibitemOpen
  \bibfield  {author} {\bibinfo {author} {\bibfnamefont {B.~L.}\ \bibnamefont
  {Higgins}}, \bibinfo {author} {\bibfnamefont {D.~W.}\ \bibnamefont {Berry}},
  \bibinfo {author} {\bibfnamefont {S.~D.}\ \bibnamefont {Bartlett}}, \bibinfo
  {author} {\bibfnamefont {H.~M.}\ \bibnamefont {Wiseman}}, \ and\ \bibinfo
  {author} {\bibfnamefont {G.~J.}\ \bibnamefont {Pryde}},\ }\href
  {http://dx.doi.org/10.1038/nature06257} {\bibfield  {journal} {\bibinfo
  {journal} {Nature}\ }\textbf {\bibinfo {volume} {450}},\ \bibinfo {pages}
  {393} (\bibinfo {year} {2007})}\BibitemShut {NoStop}%
\bibitem [{\citenamefont {Berry}\ \emph {et~al.}(2009)\citenamefont {Berry},
  \citenamefont {Higgins}, \citenamefont {Bartlett}, \citenamefont {Mitchell},
  \citenamefont {Pryde},\ and\ \citenamefont {Wiseman}}]{berry2009how}%
  \BibitemOpen
  \bibfield  {author} {\bibinfo {author} {\bibfnamefont {D.~W.}\ \bibnamefont
  {Berry}}, \bibinfo {author} {\bibfnamefont {B.~L.}\ \bibnamefont {Higgins}},
  \bibinfo {author} {\bibfnamefont {S.~D.}\ \bibnamefont {Bartlett}}, \bibinfo
  {author} {\bibfnamefont {M.~W.}\ \bibnamefont {Mitchell}}, \bibinfo {author}
  {\bibfnamefont {G.~J.}\ \bibnamefont {Pryde}}, \ and\ \bibinfo {author}
  {\bibfnamefont {H.~M.}\ \bibnamefont {Wiseman}},\ }\href
  {https://link.aps.org/doi/10.1103/PhysRevA.80.052114} {\bibfield  {journal}
  {\bibinfo  {journal} {Phys. Rev. A}\ }\textbf {\bibinfo {volume} {80}},\
  \bibinfo {pages} {052114} (\bibinfo {year} {2009})}\BibitemShut {NoStop}%
\bibitem [{\citenamefont {Kitaev}(1995)}]{kitaev1995quantum}%
  \BibitemOpen
  \bibfield  {author} {\bibinfo {author} {\bibfnamefont {A.~Y.}\ \bibnamefont
  {Kitaev}},\ }\href {http://arxiv.org/abs/quant-ph/9511026} {\bibfield
  {journal} {\bibinfo  {journal} {arXiv preprint arXiv:9511026}\ } (\bibinfo
  {year} {1995})}\BibitemShut {NoStop}%
\bibitem [{\citenamefont {Kimmel}\ \emph {et~al.}(2015)\citenamefont {Kimmel},
  \citenamefont {Low},\ and\ \citenamefont {Yoder}}]{kimmel2015robust}%
  \BibitemOpen
  \bibfield  {author} {\bibinfo {author} {\bibfnamefont {S.}~\bibnamefont
  {Kimmel}}, \bibinfo {author} {\bibfnamefont {G.~H.}\ \bibnamefont {Low}}, \
  and\ \bibinfo {author} {\bibfnamefont {T.~J.}\ \bibnamefont {Yoder}},\ }\href
  {\doibase 10.1103/PhysRevA.92.062315} {\bibfield  {journal} {\bibinfo
  {journal} {Phys. Rev. A}\ }\textbf {\bibinfo {volume} {92}},\ \bibinfo
  {pages} {062315} (\bibinfo {year} {2015})}\BibitemShut {NoStop}%
\bibitem [{\citenamefont {Svore}\ \emph {et~al.}(2014)\citenamefont {Svore},
  \citenamefont {Hastings},\ and\ \citenamefont {Freedman}}]{svore2014faster}%
  \BibitemOpen
  \bibfield  {author} {\bibinfo {author} {\bibfnamefont {K.~M.}\ \bibnamefont
  {Svore}}, \bibinfo {author} {\bibfnamefont {M.~B.}\ \bibnamefont {Hastings}},
  \ and\ \bibinfo {author} {\bibfnamefont {M.}~\bibnamefont {Freedman}},\
  }\href {http://dl.acm.org/citation.cfm?id=2600508.2600515} {\bibfield
  {journal} {\bibinfo  {journal} {Quantum Information \& Computation}\ }\textbf
  {\bibinfo {volume} {14}},\ \bibinfo {pages} {306} (\bibinfo {year}
  {2014})}\BibitemShut {NoStop}%
\bibitem [{\citenamefont {Wiebe}\ and\ \citenamefont
  {Granade}(2016)}]{wiebe2015efficient}%
  \BibitemOpen
  \bibfield  {author} {\bibinfo {author} {\bibfnamefont {N.}~\bibnamefont
  {Wiebe}}\ and\ \bibinfo {author} {\bibfnamefont {C.}~\bibnamefont
  {Granade}},\ }\href {https://link.aps.org/doi/10.1103/PhysRevLett.117.010503}
  {\bibfield  {journal} {\bibinfo  {journal} {Phys. Rev. Lett.}\ }\textbf
  {\bibinfo {volume} {117}},\ \bibinfo {pages} {010503} (\bibinfo {year}
  {2016})}\BibitemShut {NoStop}%
\bibitem [{\citenamefont {Paesani}\ \emph {et~al.}(2017)\citenamefont
  {Paesani}, \citenamefont {Gentile}, \citenamefont {Santagati}, \citenamefont
  {Wang}, \citenamefont {Wiebe}, \citenamefont {Tew}, \citenamefont {O'Brien},\
  and\ \citenamefont {Thompson}}]{paesani2017experimental}%
  \BibitemOpen
  \bibfield  {author} {\bibinfo {author} {\bibfnamefont {S.}~\bibnamefont
  {Paesani}}, \bibinfo {author} {\bibfnamefont {A.~A.}\ \bibnamefont
  {Gentile}}, \bibinfo {author} {\bibfnamefont {R.}~\bibnamefont {Santagati}},
  \bibinfo {author} {\bibfnamefont {J.}~\bibnamefont {Wang}}, \bibinfo {author}
  {\bibfnamefont {N.}~\bibnamefont {Wiebe}}, \bibinfo {author} {\bibfnamefont
  {D.~P.}\ \bibnamefont {Tew}}, \bibinfo {author} {\bibfnamefont {J.~L.}\
  \bibnamefont {O'Brien}}, \ and\ \bibinfo {author} {\bibfnamefont {M.~G.}\
  \bibnamefont {Thompson}},\ }\href {\doibase 10.1103/PhysRevLett.118.100503}
  {\bibfield  {journal} {\bibinfo  {journal} {Phys. Rev. Lett.}\ }\textbf
  {\bibinfo {volume} {118}},\ \bibinfo {pages} {100503} (\bibinfo {year}
  {2017})}\BibitemShut {NoStop}%
\bibitem [{\citenamefont {Wecker}\ \emph {et~al.}(2014)\citenamefont {Wecker},
  \citenamefont {Bauer}, \citenamefont {Clark}, \citenamefont {Hastings},\ and\
  \citenamefont {Troyer}}]{wecker2014gate}%
  \BibitemOpen
  \bibfield  {author} {\bibinfo {author} {\bibfnamefont {D.}~\bibnamefont
  {Wecker}}, \bibinfo {author} {\bibfnamefont {B.}~\bibnamefont {Bauer}},
  \bibinfo {author} {\bibfnamefont {B.~K.}\ \bibnamefont {Clark}}, \bibinfo
  {author} {\bibfnamefont {M.~B.}\ \bibnamefont {Hastings}}, \ and\ \bibinfo
  {author} {\bibfnamefont {M.}~\bibnamefont {Troyer}},\ }\href
  {http://link.aps.org/doi/10.1103/PhysRevA.90.022305} {\bibfield  {journal}
  {\bibinfo  {journal} {Phys. Rev. A}\ }\textbf {\bibinfo {volume} {90}},\
  \bibinfo {pages} {022305} (\bibinfo {year} {2014})}\BibitemShut {NoStop}%
\bibitem [{\citenamefont {Poulin}\ \emph {et~al.}(2015)\citenamefont {Poulin},
  \citenamefont {Hastings}, \citenamefont {Wecker}, \citenamefont {Wiebe},
  \citenamefont {Doherty},\ and\ \citenamefont {Troyer}}]{poulin2015trotter}%
  \BibitemOpen
  \bibfield  {author} {\bibinfo {author} {\bibfnamefont {D.}~\bibnamefont
  {Poulin}}, \bibinfo {author} {\bibfnamefont {M.~B.}\ \bibnamefont
  {Hastings}}, \bibinfo {author} {\bibfnamefont {D.}~\bibnamefont {Wecker}},
  \bibinfo {author} {\bibfnamefont {N.}~\bibnamefont {Wiebe}}, \bibinfo
  {author} {\bibfnamefont {A.~C.}\ \bibnamefont {Doherty}}, \ and\ \bibinfo
  {author} {\bibfnamefont {M.}~\bibnamefont {Troyer}},\ }\href
  {http://dl.acm.org/citation.cfm?id=2871401.2871402} {\bibfield  {journal}
  {\bibinfo  {journal} {Quantum Info. Comput.}\ }\textbf {\bibinfo {volume}
  {15}},\ \bibinfo {pages} {361} (\bibinfo {year} {2015})}\BibitemShut
  {NoStop}%
\bibitem [{\citenamefont {Berry}\ \emph {et~al.}(2015)\citenamefont {Berry},
  \citenamefont {Childs}, \citenamefont {Cleve}, \citenamefont {Kothari},\ and\
  \citenamefont {Somma}}]{berry2015simulating}%
  \BibitemOpen
  \bibfield  {author} {\bibinfo {author} {\bibfnamefont {D.~W.}\ \bibnamefont
  {Berry}}, \bibinfo {author} {\bibfnamefont {A.~M.}\ \bibnamefont {Childs}},
  \bibinfo {author} {\bibfnamefont {R.}~\bibnamefont {Cleve}}, \bibinfo
  {author} {\bibfnamefont {R.}~\bibnamefont {Kothari}}, \ and\ \bibinfo
  {author} {\bibfnamefont {R.~D.}\ \bibnamefont {Somma}},\ }\href
  {http://link.aps.org/doi/10.1103/PhysRevLett.114.090502} {\bibfield
  {journal} {\bibinfo  {journal} {Phys. Rev. Lett.}\ }\textbf {\bibinfo
  {volume} {114}},\ \bibinfo {pages} {090502} (\bibinfo {year}
  {2015})}\BibitemShut {NoStop}%
\bibitem [{\citenamefont {Low}\ and\ \citenamefont
  {Chuang}(2016)}]{low2016hamiltonian}%
  \BibitemOpen
  \bibfield  {author} {\bibinfo {author} {\bibfnamefont {G.~H.}\ \bibnamefont
  {Low}}\ and\ \bibinfo {author} {\bibfnamefont {I.~L.}\ \bibnamefont
  {Chuang}},\ }\href {http://arxiv.org/abs/1610.06546} {\bibfield  {journal}
  {\bibinfo  {journal} {arXiv preprint arXiv:1610.06546}\ } (\bibinfo {year}
  {2016})}\BibitemShut {NoStop}%
\bibitem [{\citenamefont {Peruzzo}\ \emph {et~al.}(2014)\citenamefont
  {Peruzzo}, \citenamefont {McClean}, \citenamefont {Shadbolt}, \citenamefont
  {Yung}, \citenamefont {Zhou}, \citenamefont {Love}, \citenamefont
  {Aspuru-Guzik},\ and\ \citenamefont {O'Brien}}]{peruzzo2014variational}%
  \BibitemOpen
  \bibfield  {author} {\bibinfo {author} {\bibfnamefont {A.}~\bibnamefont
  {Peruzzo}}, \bibinfo {author} {\bibfnamefont {J.}~\bibnamefont {McClean}},
  \bibinfo {author} {\bibfnamefont {P.}~\bibnamefont {Shadbolt}}, \bibinfo
  {author} {\bibfnamefont {M.-H.}\ \bibnamefont {Yung}}, \bibinfo {author}
  {\bibfnamefont {X.-Q.}\ \bibnamefont {Zhou}}, \bibinfo {author}
  {\bibfnamefont {P.~J.}\ \bibnamefont {Love}}, \bibinfo {author}
  {\bibfnamefont {A.}~\bibnamefont {Aspuru-Guzik}}, \ and\ \bibinfo {author}
  {\bibfnamefont {J.~L.}\ \bibnamefont {O'Brien}},\ }\href {\doibase
  10.1038/ncomms5213} {\bibfield  {journal} {\bibinfo  {journal} {Nature
  Communications}\ }\textbf {\bibinfo {volume} {5}},\ \bibinfo {pages} {1}
  (\bibinfo {year} {2014})}\BibitemShut {NoStop}%
\bibitem [{\citenamefont {McClean}\ \emph {et~al.}(2017)\citenamefont
  {McClean}, \citenamefont {Kivlichan}, \citenamefont {Sung}, \citenamefont
  {Steiger}, \citenamefont {Cao}, \citenamefont {Dai}, \citenamefont {Fried},
  \citenamefont {Gidney}, \citenamefont {Gimby}, \citenamefont {Gokhale},
  \citenamefont {H{\"{a}}ner}, \citenamefont {Hardikar}, \citenamefont
  {Havl{\'{i}}{\v{c}}ek}, \citenamefont {Huang}, \citenamefont {Izaac},
  \citenamefont {Jiang}, \citenamefont {Liu}, \citenamefont {Neeley},
  \citenamefont {O'Brien}, \citenamefont {Ozfidan}, \citenamefont {Radin},
  \citenamefont {Romero}, \citenamefont {Rubin}, \citenamefont {Sawaya},
  \citenamefont {Setia}, \citenamefont {Sim}, \citenamefont {Sung},
  \citenamefont {Steudtner}, \citenamefont {Sun}, \citenamefont {Sun},
  \citenamefont {Zhang},\ and\ \citenamefont {Babbush}}]{openfermion2017}%
  \BibitemOpen
  \bibfield  {author} {\bibinfo {author} {\bibfnamefont {J.~R.}\ \bibnamefont
  {McClean}}, \bibinfo {author} {\bibfnamefont {I.~D.}\ \bibnamefont
  {Kivlichan}}, \bibinfo {author} {\bibfnamefont {K.~J.}\ \bibnamefont {Sung}},
  \bibinfo {author} {\bibfnamefont {D.~S.}\ \bibnamefont {Steiger}}, \bibinfo
  {author} {\bibfnamefont {Y.}~\bibnamefont {Cao}}, \bibinfo {author}
  {\bibfnamefont {C.}~\bibnamefont {Dai}}, \bibinfo {author} {\bibfnamefont
  {E.~S.}\ \bibnamefont {Fried}}, \bibinfo {author} {\bibfnamefont
  {C.}~\bibnamefont {Gidney}}, \bibinfo {author} {\bibfnamefont
  {B.}~\bibnamefont {Gimby}}, \bibinfo {author} {\bibfnamefont
  {P.}~\bibnamefont {Gokhale}}, \bibinfo {author} {\bibfnamefont
  {T.}~\bibnamefont {H{\"{a}}ner}}, \bibinfo {author} {\bibfnamefont
  {T.}~\bibnamefont {Hardikar}}, \bibinfo {author} {\bibfnamefont
  {V.}~\bibnamefont {Havl{\'{i}}{\v{c}}ek}}, \bibinfo {author} {\bibfnamefont
  {C.}~\bibnamefont {Huang}}, \bibinfo {author} {\bibfnamefont
  {J.}~\bibnamefont {Izaac}}, \bibinfo {author} {\bibfnamefont
  {Z.}~\bibnamefont {Jiang}}, \bibinfo {author} {\bibfnamefont
  {X.}~\bibnamefont {Liu}}, \bibinfo {author} {\bibfnamefont {M.}~\bibnamefont
  {Neeley}}, \bibinfo {author} {\bibfnamefont {T.}~\bibnamefont {O'Brien}},
  \bibinfo {author} {\bibfnamefont {I.}~\bibnamefont {Ozfidan}}, \bibinfo
  {author} {\bibfnamefont {M.~D.}\ \bibnamefont {Radin}}, \bibinfo {author}
  {\bibfnamefont {J.}~\bibnamefont {Romero}}, \bibinfo {author} {\bibfnamefont
  {N.}~\bibnamefont {Rubin}}, \bibinfo {author} {\bibfnamefont {N.~P.~D.}\
  \bibnamefont {Sawaya}}, \bibinfo {author} {\bibfnamefont {K.}~\bibnamefont
  {Setia}}, \bibinfo {author} {\bibfnamefont {S.}~\bibnamefont {Sim}}, \bibinfo
  {author} {\bibfnamefont {K.}~\bibnamefont {Sung}}, \bibinfo {author}
  {\bibfnamefont {M.}~\bibnamefont {Steudtner}}, \bibinfo {author}
  {\bibfnamefont {Q.}~\bibnamefont {Sun}}, \bibinfo {author} {\bibfnamefont
  {W.}~\bibnamefont {Sun}}, \bibinfo {author} {\bibfnamefont {F.}~\bibnamefont
  {Zhang}}, \ and\ \bibinfo {author} {\bibfnamefont {R.}~\bibnamefont
  {Babbush}},\ }\href {http://arxiv.org/abs/1710.07629} {\bibfield  {journal}
  {\bibinfo  {journal} {arXiv preprint arXiv:1710.07629}\ } (\bibinfo {year}
  {2017})}\BibitemShut {NoStop}%
\bibitem [{\citenamefont {Turney}\ \emph {et~al.}(2012)\citenamefont {Turney},
  \citenamefont {Simmonett}, \citenamefont {Parrish}, \citenamefont
  {Hohenstein}, \citenamefont {Evangelista}, \citenamefont {Fermann},
  \citenamefont {Mintz}, \citenamefont {Burns}, \citenamefont {Wilke},
  \citenamefont {Abrams}, \citenamefont {Russ}, \citenamefont {Leininger},
  \citenamefont {Janssen}, \citenamefont {Seidl}, \citenamefont {Allen},
  \citenamefont {Schaefer}, \citenamefont {King}, \citenamefont {Valeev},
  \citenamefont {Sherrill},\ and\ \citenamefont {Crawford}}]{turney2012psi4}%
  \BibitemOpen
  \bibfield  {author} {\bibinfo {author} {\bibfnamefont {J.~M.}\ \bibnamefont
  {Turney}}, \bibinfo {author} {\bibfnamefont {A.~C.}\ \bibnamefont
  {Simmonett}}, \bibinfo {author} {\bibfnamefont {R.~M.}\ \bibnamefont
  {Parrish}}, \bibinfo {author} {\bibfnamefont {E.~G.}\ \bibnamefont
  {Hohenstein}}, \bibinfo {author} {\bibfnamefont {F.~A.}\ \bibnamefont
  {Evangelista}}, \bibinfo {author} {\bibfnamefont {J.~T.}\ \bibnamefont
  {Fermann}}, \bibinfo {author} {\bibfnamefont {B.~J.}\ \bibnamefont {Mintz}},
  \bibinfo {author} {\bibfnamefont {L.~A.}\ \bibnamefont {Burns}}, \bibinfo
  {author} {\bibfnamefont {J.~J.}\ \bibnamefont {Wilke}}, \bibinfo {author}
  {\bibfnamefont {M.~L.}\ \bibnamefont {Abrams}}, \bibinfo {author}
  {\bibfnamefont {N.~J.}\ \bibnamefont {Russ}}, \bibinfo {author}
  {\bibfnamefont {M.~L.}\ \bibnamefont {Leininger}}, \bibinfo {author}
  {\bibfnamefont {C.~L.}\ \bibnamefont {Janssen}}, \bibinfo {author}
  {\bibfnamefont {E.~T.}\ \bibnamefont {Seidl}}, \bibinfo {author}
  {\bibfnamefont {W.~D.}\ \bibnamefont {Allen}}, \bibinfo {author}
  {\bibfnamefont {H.~F.}\ \bibnamefont {Schaefer}}, \bibinfo {author}
  {\bibfnamefont {R.~A.}\ \bibnamefont {King}}, \bibinfo {author}
  {\bibfnamefont {E.~F.}\ \bibnamefont {Valeev}}, \bibinfo {author}
  {\bibfnamefont {C.~D.}\ \bibnamefont {Sherrill}}, \ and\ \bibinfo {author}
  {\bibfnamefont {T.~D.}\ \bibnamefont {Crawford}},\ }\href
  {https://onlinelibrary.wiley.com/doi/abs/10.1002/wcms.93} {\bibfield
  {journal} {\bibinfo  {journal} {Wiley Interdisciplinary Reviews:
  Computational Molecular Science}\ }\textbf {\bibinfo {volume} {2}},\ \bibinfo
  {pages} {556} (\bibinfo {year} {2012})}\BibitemShut {NoStop}%
\end{thebibliography}%

\appendix
\section{Errors in likelihood function}

\subsection{Subsampling Hamiltonians}
We first consider the case where terms are sampled uniformly from the Hamiltonian. Let the Hamiltonian be a sum of $L$ simulable Hamiltonians $H_\ell$, $H = \sum_{\ell=1}^L H_\ell$. Throughout we consider an eigenstate $\ket\psi$ of $H$ and its corresponding eigenenergy $E$. From the original, we can construct a new Hamiltonian
\be
H_\text{est} = \frac{L}{m} \sum_{i=1}^m H_{\ell_i}
\ee
by uniformly sampling terms $H_{\ell_i}$ from the original Hamiltonian.

When one randomly subsamples the Hamiltonian, errors are naturally introduced.  The main question is less about how large these errors are, but instead about how they impact the iterative phase estimation protocol.  The following lemma states that the impact on the likelihood functions can be made arbitrarily small.
\begin{lemma}\label{lem:subsample}
Let $\ell_i$ be an indexed family of sequences mapping $\{1,\ldots, m\}\rightarrow \{1,\ldots, L\}$ formed by uniformly sampling elements from $\{1,\ldots, L\}$ independently with replacement, and let $\{H_{\ell}: \ell=1,\ldots, L\}$ be a corresponding family of Hamiltonians with $H=\sum_{\ell=1}^L H_{\ell}$.  For $\ket{\psi}$ an eigenstate of $H$ such that $H\ket{\psi}=E\ket{\psi}$ and $H_\mathrm{samp} = \frac{L}{m} \sum_{k=1}^m H_{\ell_i(k)}$ with corresponding eigenstate $H_\mathrm{samp} \ket{\psi_i} = E_i \ket{\psi_i}$ we then have that the error in the likelihood function for phase estimation vanishes with high probability over $H_\mathrm{samp}$  in the limit of large $m$:
\begin{align*}
\left|P(o | Et; M, \theta) - P(o|E_i t;M,\theta) \right| \in O\left(\frac{MtL}{\sqrt{m}}\sqrt{\mathbb{V}_{\ell}(\bra{\psi}H_{\ell} \ket{\psi})} \right)
\end{align*}
\end{lemma}
\begin{proof}
Because the terms $H_{\ell_i(k)}$
are uniformly sampled, each set of terms $\{\ell_i\}$ is equally likely, and by linearity of expectation
$
\mathbb{E} \left[ H_\text{samp} \right] = H,
$
from which we know that $\mathbb{E}_{\{i\}} \left[ \bra{\psi} H - H_\text{est} \ket{\psi} \right] = 0$.

The second moment is easy to compute from the independence property of the distribution:
\begin{align}
\mathbb{V}_{\{i\}}(\bra{\psi} H_\mathrm{samp} \ket{\psi}) &= \frac{L^2}{m^2}\mathbb{V}_{\{i\}}(\bra{\psi} \sum_{k=1}^m H_{\ell_i(k)} \ket{\psi})=\frac{L^2}{m^2}\sum_{k=1}^m \mathbb{V}_{\{i\}}(\bra{\psi}  H_{\ell_i(k)} \ket{\psi}).\label{eq:varianceeq}
\end{align}
Since the different sequences $\ell_i$ are chosen uniformly at random, the result follows from the observation that $\mathbb{V}_{\{i\}}(\bra{\psi}  H_{\ell_i(k)} \ket{\psi}) = \mathbb{V}_{\ell}(\bra{\psi}  H_{\ell} \ket{\psi})$.

From first order perturbation theory, we have that the leading order shift in any eigenvalue is $O(\bra{\psi} (H-H_{\rm samp}) \ket{\psi})$ to within error $O(L/\sqrt{m})$.
Thus from~\eq{varianceeq} the variance in this shift is
\be
\mathbb{V}_{\{i\}}(\bra{\psi}(H-H_{\rm samp}) \ket{\psi})=\frac{L^2}{m}\mathbb{V}_{\ell}(\bra{\psi}H_{\ell} \ket{\psi}).
\ee
This further implies that the (normalized) perturbed eigenstate $\ket{\psi_i}$ has eigenvalue 
\be H_\mathrm{samp} \ket{\psi_i} = E\ket{\psi_i} + O\left(\frac{L}{\sqrt{m}}\sqrt{\mathbb{V}_{\ell}(\bra{\psi}H_{\ell} \ket{\psi})}\right)\ee with high probability over $i$ from Markov's inequality.
It then follows from Taylor's theorem and~\eq{likelihood} that
\be
\left|P(o | Et; M, \theta) - P(o|E_i t;M,\theta) \right| \in O\left(\frac{MtL}{\sqrt{m}}\sqrt{\mathbb{V}_{\ell}(\bra{\psi}H_{\ell} \ket{\psi})} \right),
\ee
with high probability over $i$.

\end{proof}

This result shows that if we sample the coefficients of the Hamiltonian that are to be included in the subsampled Hamiltonian uniformly then we can make the error in the estimate of the Hamiltonian arbitrarily small.  In this context, taking $m\rightarrow \infty$ does not cause the cost of simulation to diverge (as it would for many sampling problems).  This is because once every possible term is included in the Hamiltonian, there is no point in subsampling, and it is more efficient to directly take $H = H_\mathrm{samp}$ to eliminate the variance in the likelihood function that would arise from subsampling.  In general we need to take $m\in \Omega(\mathbb{V}_{\ell}(\bra{\psi}H_{\ell} \ket{\psi})/(MtL)^2)$ in order to guarantee that the error in the likelihood function is at most a constant.  Thus this shows that as any iterative phase estimation algorithm proceeds, that (barring the problem of accidentally exciting a state due to perturbation) we will be able to find a good estimate of the eigenphase by taking $m$ to scale inverse-quadratically with $M$.

\section{Bayesian phase estimation using random Hamiltonians}

\begin{theorem}
Let $E$ be an event and let $P(E|\theta)$ and $P'(E|\theta)$ for $\theta \in [-\pi,\pi)$ be two likelihood functions such that $\max_\theta(|P(E|\theta) - P'(E|\theta)|)\le \Delta$, and further assume that for prior $P(\theta)$ we have that $\min(P(E),P'(E))\ge 2\Delta$.  We then have that
$$
\left|\int \theta\big(P(\theta|E)-P'(\theta|E)\big) \mathrm{d}\theta \right|\le \frac{5\pi \Delta}{P(E)}.
$$
Further, if $P(E|\theta)= \prod_{j=1}^N P(E_j|\theta)$ with $1-|P'(E_j|\theta)/P(E_j|\theta)|\le \gamma$ then
$$
\left|\int \theta\big(P(\theta|E)-P'(\theta|E)\big) \mathrm{d}\theta \right|\le {5\pi ((1+\gamma)^N-1)}.
$$
\end{theorem}
\begin{proof}
From the triangle inequality we have that
\begin{equation}
|P(E)-P'(E)| = \left|\int P(\theta) \big(P(E|\theta) - P'(E|\theta)\big) \mathrm{d}\theta\right|\le \Delta.
\end{equation}
Thus it follows from the assumption that $P'(E)\ge 2\Delta$ that
\begin{align}
\left|P(\theta|E) - P'(\theta|E)\right| &= P(\theta)\left|\frac{P(E|\theta)}{P(E)} - \frac{P'(E|\theta)}{P'(E)} \right| \nonumber\\
& \le P(\theta) \left(\left|\frac{P(E|\theta)}{P(E)} - \frac{P'(E|\theta)}{P(E)} \right| + \left|\frac{P'(E|\theta)}{P(E)} - \frac{P'(E|\theta)}{P'(E)} \right|\right)\nonumber\\
&\le P(\theta) \left(\frac{\Delta}{P(E)} + \left|\frac{P'(E|\theta)}{P'(E)-\Delta} - \frac{P'(E|\theta)}{P'(E)} \right|\right)\nonumber\\
&\le \Delta \left(\frac{P(\theta)}{P(E)} + \frac{2P'(E|\theta)}{P'^2(E)}\right).
\end{align}
Thus we have that
\begin{align}
\left|\int \theta(P(\theta|E)-P'(\theta|E)) \mathrm{d}\theta \right| &\le \Delta \left(\frac{\langle \theta \rangle_{prior}}{P(E)} + \frac{2\langle \theta \rangle'_{posterior}}{P'(E)} \right),\nonumber\\
&\le\pi\Delta \left(\frac{1}{P(E)} + \frac{2}{P'(E)} \right)\nonumber\\
&\le\pi\Delta \left(\frac{1}{P(E)} + \frac{2}{P(E)-\Delta} \right)\nonumber\\
&\le \frac{5\pi\Delta}{P(E)}.
\end{align}
Now if we assume that we that we have a likelihood function that factorizes over $N$ experiments, we can take
\begin{equation}
\frac{\Delta}{P(E)} = \frac{\prod_{j=1}^NP(E_j|\theta) - \prod_{j=1}^N P'(E_j|\theta)}{\prod_{j=1}^NP(E_j|\theta) }=\prod_{j=1}^N1 - \prod_{j=1}^N \frac{P'(E_j|\theta)}{P(E_j|\theta)}.
\end{equation}
From the triangle inequality
\begin{equation}
\prod_{j=1}^N1 - \prod_{j=1}^N \frac{P'(E_j|\theta)}{P(E_j|\theta)} \le \left|\prod_{j=1}^{N-1}1 - \prod_{j=1}^{N-1} \frac{P'(E_j|\theta)}{P(E_j|\theta)}\right| + \left|1-\frac{P'(E_N|\theta)}{P(E_N|\theta)}\right|(1+\gamma)^N\le \left|\prod_{j=1}^{N-1}1 - \prod_{j=1}^{N-1} \frac{P'(E_j|\theta)}{P(E_j|\theta)}\right| +\gamma (1+\gamma)^N.
\end{equation}
Solving this recurrence relation gives
\begin{equation}
\prod_{j=1}^N1 - \prod_{j=1}^N {P'(E_j|\theta)}{P(E_j|\theta)} \le (1+\gamma)^N-1.
\end{equation}
Thus the result follows.
\end{proof}
\section{Shift in the posterior mean from using random Hamiltonians}

We will analyze the shift in the posterior mean of the estimated phase assuming a random shift $\delta(\phi)$ in the joint likelihood of all the experiments,
\begin{equation}
P^\prime(\vec o | \phi; \vec M, \vec \theta) = P(\vec o | \phi; \vec M, \vec \theta) + \delta(\phi).
\end{equation}
Here, $P(\vec o | \phi; \vec M, \vec \theta)$ is the joint likelihood of a series of $N$ outcomes $\vec o$ given a true phase $\phi$ and the experimental parameters $\vec M$ and $\vec\theta$ for the original Hamiltonian. $P^\prime(\vec o | \phi; \vec M, \vec \theta)$ is the joint likelihood with a new random Hamiltonian in each experiment. By a vector like $\vec M$, we mean the repetitions for each experiment performed in the series; $M_i$ is the number of repetitions in the $i^\text{th}$ experiment.

First, we will work backward from the assumption that the joint likelihood is shifted by some amount $\delta(\phi)$, to determine an upper bound on the acceptable difference in ground state energies between the true and the random Hamiltonians. 
We will do this by working backwards from the shift in the joint likelihood of all experiments, to the shifts in the likelihoods of individual experiments, and finally to the corresponding tolerable differences between the ground state energies. 
Second, we will use this result to determine the shift in the posterior mean in terms of the differences in energies, as well as its standard deviation over the ensemble of randomly generated Hamiltonians.

\subsection{Shifts in the joint likelihood}
\label{sec:likelihoodshifts}

The random Hamiltonians for each experiment lead to a random shift in the joint likelihood of a series of outcomes
\begin{equation}
P^\prime(\vec o | \phi; \vec M, \vec \theta) = P(\vec o | \phi; \vec M, \vec \theta) + \delta(\phi).
\end{equation}
We would like to determine the maximum possible change in the posterior mean under this shifted likelihood. We will work under the assumption that the mean shift in the likelihood over the prior is at most $|\bar \delta| \le P(\vec o) / 2$. The posterior is
\begin{equation}
\begin{aligned}
P'(\phi | \vec o; \vec M, \vec\theta) =& \frac{P'(\vec o | \phi; \vec M, \vec\theta) P(\phi) }{ \int P'(\vec o | \phi; \vec M, \vec\theta) P(\phi) \,d\phi } \\
=& \frac{P(\vec o | \phi; \vec M, \vec\theta) P(\phi) + \delta(\phi) P(\phi) }{ \int (P(\vec o | \phi; \vec M, \vec\theta) P(\phi) + \delta(\phi) P(\phi)) \,d\phi } \\
=& \frac{P(\vec o | \phi; \vec M, \vec\theta) P(\theta) + \delta(\phi) P(\phi) }{ P(\vec o) + \bar \delta }.
\end{aligned}
\end{equation}
We can make progress toward bounding the shift in the posterior by first bounding the shift in the joint likelihood in terms of the shifts in the likelihoods of the individual experiments, as follows.

\begin{lemma}
Let $P( o_j | \phi; M_j, \theta_j)$ be the likelihood of outcome $o_j$ on the $j^\text{th}$ experiment for the Hamiltonian $H$, and $P^\prime( o_j | \phi; M_j, \theta_j) = P( o_j | \phi; M_j, \theta_j) + \epsilon_j(\phi)$ be the likelihood with the randomly generated Hamiltonian $H_j$. Assume that $N \max_j (|\epsilon_j(\phi)| / P(o_j|\phi, M_j, \theta_j)) < 1$ and $|\epsilon_j(\phi)| \le P(o_j|\phi, M_j, \theta_j) / 2$ for all experiments $j$. Then the mean shift in the joint likelihood of all $N$ experiments,
$$|\bar\delta| = \left| \int P(\phi)\left( P^\prime(\vec o | \phi; \vec M, \vec \theta) - P(\vec o | \phi; \vec M, \vec \theta) \right) \,d\phi \right|,
$$
is at most $|\bar\delta| \le  2 \sum_{j=1}^N \max_\phi \frac{|\epsilon_j(\phi)|}{P(o_j|\phi; M_j, \theta_j)} P(\vec o)$.
\label{lem:likelihoodshift}
\end{lemma}
\begin{proof}
We can write the joint likelihood in terms of the shift $\epsilon_j(\phi)$ to the likelihoods of each of the $N$ experiments in the sequence, $P^\prime(o_j | \phi; M_j, \theta_j) = P(o_j | \phi; M_j, \theta_j) + \epsilon_j(\phi)$. The joint likelihood is $P'(\vec o|\phi; \vec M, \vec\theta) = \prod_{j=1}^N \left( P(o_j | \phi; M_j, \theta_j) + \epsilon_j(\phi) \right)$, so
\begin{equation}
\begin{aligned}
\log P'(\vec o|\phi; \vec M, \vec\theta) =& \log\left(\prod_{j=1}^N \left( P(o_j | \phi; M_j, \theta_j) + \epsilon_j(\phi) \right) \right) \\
=& \sum_{j=1}^N \left[ \log P(o_j|\phi, M_j, \theta_j) + \log\left(1 + \frac{\epsilon_j(\phi)}{P(o_j|\phi, M_j, \theta_j)}\right) \right] \\
=& \log P(\vec o|\phi; \vec M, \vec\theta) + \sum_{j=1}^N \log\left(1 + \frac{\epsilon_j(\phi)}{P(o_j|\phi, M_j, \theta_j)}\right)
\end{aligned}
\end{equation}
This gives us the ratio of the shifted to the unshifted joint likelihood,
\begin{equation}
\begin{aligned}
\frac{P'(\vec o|\phi; \vec M, \vec\theta)}{P(\vec o|\phi; \vec M, \vec\theta)} =& \exp\left[ \sum_{j=1}^N \log\left(1 + \frac{\epsilon_j(\phi)}{P(o_j|\phi, M_j, \theta_j)}\right) \right].
\end{aligned}
\end{equation}
We will linearize and simplify this using inequalities for the logarithm and exponential. By finding inequalities which either upper or lower bound both these functions, we can upper bound $|\delta(\phi)|$ in terms of the unshifted likelihoods $P(\vec o|\phi; \vec M, \vec\theta)$ and $P( o_j |\phi; M_j, \theta_j)$, and the shift in the single-experiment likelihood $\epsilon_j(\phi)$. 

The inequalities we will use to sandwich the ratio are
\begin{equation}
\begin{aligned}
& 1 - |x| \le \exp(x), && \text{ for $x \le 0$ } \\
& \exp(x) \le 1 + 2|x|, && \text{ for $x < 1$ } \\
& -2|x| \le \log(1+x), && \text{ for $|x| \le 1/2$ } \\
& \log(1+x) \le |x|, && \text{ for $x \in \mathbb{R}$ }
\end{aligned}
\end{equation}
In order for all four inequalities to hold, we must have that $N \max_j (|\epsilon_j(\phi)| / P(o_j|\phi, M_j, \theta_j)) < 1$ (for the exponential inequalities) and $|\epsilon_j(\phi)| \le P(o_j|\phi, M_j, \theta_j) / 2$ for all $j$ (for the logarithm inequalities). Using them to upper bound the ratio of the shifted to the unshifted likelihood,
\begin{equation}
\begin{aligned}
\frac{P'(\vec o|\phi; \vec M, \vec\theta)}{P(\vec o|\phi; \vec M, \vec\theta)} = \frac{P(\vec o|\phi; \vec M, \vec\theta) + \delta(\phi)}{P(\vec o|\phi; \vec M, \vec\theta)} \le& \exp\left[ \sum_{j=1}^N \frac{|\epsilon_j(\phi)|}{P(o_j|\phi, M_j, \theta_j)} \right] \le 1 + 2 \sum_{j=1}^N \frac{|\epsilon_j(\phi)|}{P(o_j|\phi, M_j, \theta_j)}; \\
\delta(\phi) \le& 2 P(\vec o|\phi; \vec M, \vec\theta) \sum_{j=1}^N \frac{ |\epsilon_j(\phi)| }{P(o_j|\phi, M_j, \theta_j)} .
\end{aligned}
\end{equation}
On the other hand, using them to lower bound the ratio,
\begin{equation}
\begin{aligned}
\frac{P'(\vec o|\phi; \vec M, \vec\theta)}{P(\vec o|\phi; \vec M, \vec\theta)} = \frac{P(\vec o|\phi; \vec M, \vec\theta) + \delta(\phi)}{P(\vec o|\phi; \vec M, \vec\theta)} \ge& \exp\left[ -2 \sum_{j=1}^N \frac{|\epsilon_j(\phi)|}{P(o_j|\phi; M_j, \theta_j)} \right] \ge 1 - 2 \sum_{j=1}^N \frac{|\epsilon_j(\phi)|}{P(o_j|\phi; M_j, \theta_j)}; \\
\delta(\phi) \ge& -2 P(\vec o|\phi; \vec M, \vec\theta) \sum_{j=1}^N \frac{ |\epsilon_j(\phi)| }{P(o_j|\phi; M_j, \theta_j)} .
\end{aligned}
\end{equation}
The upper and lower bounds are identical up to sign. This allows us to combine them directly, so we have
\begin{equation}
|\delta(\phi)| \le 2 P(\vec o|\phi; \vec M, \vec\theta) \sum_{j=1}^N \frac{|\epsilon_j(\phi)|}{P(o_j|\phi; M_j, \theta_j)}.
\end{equation}
From this, we find our upper bound on the mean shift over the posterior, $|\bar\delta|$, since by the triangle inequality
\begin{equation}
\begin{aligned}
|\bar\delta| = \left| \int \delta(\phi) P(\phi) \,d\phi \right| \le& 2 \sum_{j=1}^N \int \left( \frac{|\epsilon_j(\phi)|}{P(o_j|\phi; M_j, \theta_j)} P(\vec o|\phi; \vec M, \vec\theta) P(\phi) \right) \,d\phi \\
\le&  2 \sum_{j=1}^N \max_\phi \frac{|\epsilon_j(\phi)|}{P(o_j|\phi; M_j, \theta_j)} P(\vec o).
\end{aligned}
\end{equation}
\end{proof}
So we have a bound on the shift in the joint likelihood in terms of the shifts in the likelihoods of individual experiments. These results allow us to bound the shift in the posterior mean in terms of the shifts in the likelihoods of the individual experiments $\epsilon_j(\phi)$.

\subsection{Shift in the posterior mean}
We make use of the assumption that $|\bar\delta| \le P(\vec o)/2$ to bound the shift in the posterior mean. 

\begin{lemma}
Assuming in addition to the assumptions of \lem{likelihoodshift} that $|\bar\delta| \le P(\vec o)/2$, the difference between the the posterior mean that one would see with the ideal likelihood function and the perturbed likelihood function is at most $$|\bar\phi - \bar\phi^\prime| \le 8\max_\phi \left(\sum_{j=1}^N \frac{|\epsilon_j(\phi)|}{P(o_j|\phi, M_j, \theta_j)}\right) \overline{|\phi|}^\mathrm{post}. $$
\label{lem:posteriorshift}
\end{lemma}
\begin{proof}
We approach the problem of bounding the difference between the posterior means by bounding the point-wise difference between the shifted posterior and the posterior with the original Hamiltonian,
\begin{equation}
\begin{aligned}
| P(\phi | \vec o; \vec M, \vec\theta) - P'(\phi | \vec o; \vec M, \vec\theta) | = \left| \frac{P(\vec o|\phi; \vec M, \vec\theta) P(\phi)}{P(\vec o)} - \frac{P(\vec o|\phi; \vec M, \vec\theta) P(\phi) + \delta(\phi) P(\phi) }{P(\vec o) + \bar \delta} \right|.
\end{aligned}
\end{equation}
As a first step, we place an upper bound on the denominator of the shifted posterior, $(P(\vec o) + \bar\delta)^{-1}$:
\begin{equation}
\begin{aligned}
\frac1{P(\vec o) + \bar\delta} 
=& \frac{1}{P(\vec o)} \sum_{k=0}^\infty \left( \frac{-\bar\delta}{P(\vec o)} \right)^k \\
=& \frac{1}{P(\vec o)} - \frac{\bar\delta}{P(\vec o)^2} + \frac{\bar\delta ^2}{P(\vec o)^3} \sum_{k=0}^\infty \left( \frac{-\bar\delta}{P(\vec o)} \right)^k \\
\le& \frac{1}{P(\vec o)} + \frac{2 |\bar\delta|}{P(\vec o)^2} = \frac{1 + 2|\bar\delta| / P(\vec o)}{P(\vec o)} ,
\end{aligned}
\end{equation}
where in the two inequalities we used the assumption that $|\bar \delta| \le P(\vec o) / 2$. Using this, the point-wise difference between the posteriors is at most
\begin{equation}
\begin{aligned}
\left| \frac{P(\vec o|\phi; \vec M, \vec\theta) P(\phi)}{P(\vec o)} - \frac{P(\vec o|\phi; \vec M, \vec\theta) P(\phi) + \delta(\theta) P(\theta) }{P(\vec o) + \bar \delta} \right| \le& \left| \frac{P(\vec o|\phi; \vec M, \vec\theta) P(\phi)}{P(\vec o)} - \frac{P(\vec o|\phi; \vec M, \vec\theta) P(\phi) }{P(\vec o) + \bar \delta}\right| + \left|\frac{\delta(\phi) P(\phi) }{P(\vec o) + \bar \delta} \right| \\
\le& \frac{2 |\bar\delta| P(\vec o|\phi; \vec M, \vec\theta) P(\phi) }{P(\vec o)^2} + \frac{|\delta(\phi)| P(\phi)}{P(\vec o)} \left( 1 + \frac{2 |\bar\delta|}{P(\vec o)} \right) \\
\le& \frac{2 |\bar\delta| P(\vec o|\phi; \vec M, \vec\theta) P(\phi) }{P(\vec o)^2} + \frac{2 |\delta(\phi)| P(\phi)}{P(\vec o)},
\end{aligned}
\end{equation}
again using $|\bar\delta| \le P(\vec o)/2$. With this we can bound the change in the posterior mean,
\begin{equation}
\begin{aligned}
\label{eq:posteriormean1}
|\bar\phi - \bar{\phi^\prime}| \le& \int |\phi| |P(\phi | \vec o; \vec M, \vec\theta) - P^\prime(\phi | \vec o; \vec M, \vec\theta)| \, d\phi \\
\le& \frac{2}{P(\vec o)} \int |\phi| \left( \frac{ |\bar\delta| P(\vec o|\phi; \vec M, \vec\theta) P(\phi) }{P(\vec o)} + |\delta(\phi)| P(\phi) \right) \,d\phi \\
\le& \frac{2}{P(\vec o)} \int |\phi| |\delta(\phi)| P(\phi) \,d\phi + \frac{2 |\bar\delta|}{P(\vec o)} \int |\phi| \left( \frac{ P(\vec o|\phi; \vec M, \vec\theta) P(\phi) }{P(\vec o)} \right) \,d\phi \\
\le& \frac{2}{P(\vec o)} \left( \int |\phi| |\delta(\phi)| P(\phi) \,d\phi + \overline{|\phi|}^\text{post} |\bar\delta| \right)
\end{aligned}
\end{equation}

Now, our bounds from \lem{likelihoodshift} allow us to bound the shift on the posterior mean in terms of the shifts in the likelihoods of individual experiments, $\epsilon_j(\phi)$,
\begin{equation}
\begin{aligned}
\label{eqn:posteriormean2}
|\bar\phi - \bar{\phi^\prime}| \le& \frac{2}{P(\vec o)} \left( \int |\phi| |\delta(\phi)| P(\phi) \,d\phi + \overline{|\phi|}^\text{post} |\bar\delta| \right) \\
\le& \frac{2}{P(\vec o)} \left( 2 \max_\phi \sum_{j=1}^N \frac{|\epsilon_j(\phi)|}{P(o_j|\phi, M_j, \theta_j)} P(\vec o) \int |\phi| \left( \frac{ P(\vec o|\phi; \vec M, \vec\theta) P(\phi) }{P(\vec o)} \right) \,d\phi + \overline{|\phi|}^\text{post} |\bar\delta| \right),
\end{aligned}
\end{equation}
where in the last step we multiplied and divided by $P(\vec o)$. This is
\begin{equation}
\begin{aligned}
|\bar\phi - \bar{\phi^\prime}| \le& \frac{2}{P(\vec o)} \left( 2 \max_\phi \left( \sum_{j=1}^N \frac{|\epsilon_j(\phi)|}{P(o_j|\phi, M_j, \theta_j)} \right) P(\vec o) \overline{|\phi|}^\text{post} + \overline{|\phi|}^\text{post} |\bar\delta| \right) \\
\le& 8 \max_\phi \left(\sum_{j=1}^N \frac{|\epsilon_j(\phi)|}{P(o_j|\phi, M_j, \theta_j)}\right) \overline{|\phi|}^\text{post}.
\end{aligned}
\end{equation}
\end{proof}

\subsection{Acceptable shifts in the phase}
The final question we are interested in is what the bound on the shift in the posterior mean is in terms of shifts in the phase.

\begin{theorem}
If the assumptions of~\lem{posteriorshift} hold, for all $j$ and $x\in [-\pi,\pi)$ $P(o_j|\theta; x, \theta_j)=\frac{1 + (-1)^{o_j} \cos(M_j(\theta_j - x))}{2}$, for each of the $N$ experiments we have that the eigenphases $\{\phi_j': j=1,\ldots N\}$ used in phase estimation and the eigenphase of the true Hamiltonian $\phi$ obey $|\phi - \phi'_j|\le |\Delta \phi |$, and additionally $P(o_j|\phi,M_j,\theta_j) \in \Theta(1)$ then we have that the shift in the posterior mean of the eigenphase that arises from inaccuracies in the eigenvalues in the intervening Hamiltonians obeys
$$|\bar\phi - \bar{\phi^\prime}| \le 8 \pi \max_\phi \left(\sum_{j=1}^N \frac{M_j}{P(o_j|\phi; M_j, \theta_j)}\right) |\Delta\phi|.$$
Furthermore, if $\sum_j M_j \in O(1/\epsilon_\phi)$ and $P(o_j|\phi; M_j, \theta_j) \in \Theta(1)$ for all $j$, then
$$
|\phi - \bar{\phi^\prime}|  \in O\left(\frac{|\Delta \phi|}{\epsilon_\phi} \right).
$$ 
\end{theorem}
\begin{proof}
We can express the shift in the posterior mean in terms of the shift in the phase applied to the ground state, $\Delta\phi$, by bounding $\epsilon_j(\phi)$ in terms of it. Recall that the likelihood with the random Hamiltonian is
\begin{equation}
P^\prime(o_j | \phi; M_j, \theta_j) = P(o_j | \phi; M_j, \theta_j) + \epsilon_j(\phi),
\end{equation}
where the unshifted likelihood for the $j^\text{th}$ experiment is $P(o_j | \phi; M_j, \theta_j) = \frac{1}{2} \left(1 + (-1)^{o_j} \cos( M_j(\phi - \theta_j) \right)$. Thus,
\begin{equation}
\begin{aligned}
|\epsilon_j(\phi)| =& \frac{1}{2} \left|\cos( M_j(\phi + \Delta\phi - \theta_j) - \cos( M_j(\phi - \theta_j) \right| \le  M_j |\Delta\phi|,
\end{aligned}
\end{equation}
using the upper bound on the derivative $\sin(x) \le |x|$. In sum, we have that the error in the posterior mean is at most
\begin{equation}
|\bar\phi - \bar{\phi^\prime}| \le 8 \max_\phi \left(\sum_{j=1}^N \frac{M_j}{P(o_j|\phi, M_j, \theta_j)}\right) \overline{|\phi|}^\text{post} |\Delta\phi |.
\end{equation}
The result then follows from the fact that the absolute value of the posterior mean is at most $\pi$ if the branch $[-\pi,\pi)$ is chosen.
\end{proof}

\section{Shift in the eigenphase with a new random Hamiltonian in each repetition}
\label{app:phaseshift}

We can reduce the variance in the applied phase by generating a different Hamiltonian in each repetition. However, this comes at a cost: we can view this cost either as leading to a failure probability in the evolution, or more generally to an additional phase shift.

The reason this reduces the variance is that the phase across repetitions is uncorrelated. Instead of just having a single Hamiltonian in the variance, the variance is over the indices of $\lceil M_j \rceil$ different Hamiltonians. 
Because of this, the variance only scales as $M\mathbb{V}[\phi_\text{est}]$ instead of $M^2\mathbb{V}[\phi_\text{est}]$ as it usually would (from the underlying variance in $\phi_\text{est}$). The cost is that, by reducing the variance in the phase in this way, we introduce an additional shift in the phase. 
Were we not to resample across multiple steps, we would have the same ground state through the entire process with the same wrong Hamiltonian. Instead, resampling means we only approximately have the same ground state, at the cost of the variance being lower by a factor $M_j$. 
Since the additional shift is also linear in $\lceil M_j \rceil$, this can lead to an improvement. It generally requires that the gap be small, and that $\lambda_j \propto \| H_j - H_{j-1}\|$ be small.

We then have a tradeoff between having the variance scale as $M_j^2 \mathbb{V}[\phi_\text{est}]$, or ${M_j} \mathbb{V}[\phi_\text{est}]$ with this new shift which scales linearly with $M_j$. We work to better understand this tradeoff in the following sections.

\subsection{Failure probability of the algorithm}

For phase estimation, we can reduce the variance of the estimate in the phase by randomizing within the repetitions for each experiment. Let us say the $j^\text{th}$ experiment has $M_j$ repetitions. 

Within each repetition, we randomly generate a new Hamiltonian $H_{k}$. Each Hamiltonian $H_k$ has a slightly different ground state and energy. 
The reason this resampling reduces the variance in the estimated phase is that the phases between repetitions are uncorrelated: whereas for the single-Hamiltonian case, the variance in the phase $\exp(-iM\phi_\text{est})$ is $\mathbb{V}[M\phi_\text{est}] = M^2 \mathbb{V}[\phi_\text{est}]$, when we simulate a different random Hamiltonian in each repetition (and estimate the sum of the phases, as $\exp(-i\sum_{k=1}^M \phi_{k,\text{est}})$), the variance is $\mathbb{V}[\sum_{k=1}^M \phi_{k,\text{est}}] = \sum_{k=1}^M \mathbb{V}[\phi_{k,\text{est}}] = M \mathbb{V}[\phi_\text{est}]$.

By evolving under a different random instantiation of the Hamiltonian in each repetition, the variance in the phase is quadratically reduced; the only cost is that the algorithm now has either a failure probability (of leaving the ground state from repetition to repetition, i.e.\ in the transition from the ground state of $H_{k-1}$ to the ground state of $H_k$) or an additional phase shift compared to the true sum of the ground state energies. 
The first case is simpler to analyze: we show in \lem{failureprob}, provided that the gap is sufficiently small, that the failure probability can be made arbitrarily small. We do this by viewing the success probability of the algorithm as the probability of remaining in the ground state throughout the sequence of $\lceil M_j \rceil$ random Hamiltonians. In the second case, we prove in \lem{phaseUk} a bound on the difference between eigenvalues if the state only leaves the ground space for short intervals during the evolution.

\begin{lemma}
Consider a sequence of Hamiltonians $\{ H_k \}_{k=1}^M$, $M>1$.
Let $\gamma$ be the minimum gap between the ground and first excited energies of any of the Hamiltonians, $\gamma = \min_k (E_1^k - E_0^k)$. Similarly, let $\lambda = \max_k \|H_k - H_{k-1}\|$ be the maximum difference between any two in the sequence. The probability of leaving the ground state when transferring from $H_1$ to $H_2$ through to $H_M$ in order is at most $0 < \epsilon < 1$ provided that
$$\frac{ \lambda }{ \gamma } < \sqrt{1 - \exp\left( \frac{\log(1-\epsilon)}{ M-1 } \right)}.$$
\label{lem:failureprob}
\end{lemma}
\begin{proof}
Let $\ket{\psi_i^k}$ be the $i^\text{th}$ eigenstate of the Hamiltonian $H_k$ and let $E_i^k$ be the corresponding energy. Given that the algorithm begins in the ground state of $H_1$ (the state $\ket{\psi_0^1}$), the probability of remaining in the ground state through all $M$ steps is
\begin{equation}
\left| \braket{\psi^{M}_0}{\psi^{M - 1}_0} \cdots \braket{\psi^{2}_0}{\psi^{1}_0} \right|^2.
\end{equation}
This is the probability of the algorithm staying in the ground state in every segment. We can simplify this expression by finding a bound for $| \ket{\psi_0^k} - \ket{\psi_0^{k - 1}} |^2$. Let $\lambda_k V_k = H_{k} - H_{k-1}$, where we choose $\lambda_k$ such that $\| V_k\| = 1$ to simplify our proof. Treating $\lambda_k V_k$ as a perturbation on $H_{k-1}$, the components of the shift in the ground state of $H_{k-1}$ are bounded by the derivative
\begin{equation}
\left| \frac{\partial}{\partial \lambda} \braket{\psi_0^{k-1}}{\psi_\ell^{k}} \right| = \frac{| \bra{\psi_\ell^{k-1}} V_k \ket{\psi_0^{k-1}} |}{E_\ell^{k-1} - E_0^{k-1}}
\end{equation}
multiplied by $\lambda = \max |\lambda_k|$, where the maximization is over both $k$ as well as perturbations for a given $k$. Using this,
\begin{equation}
\begin{aligned}
|\braket{\psi_{\ell}^{k}}{\psi_0^{k-1}}|^2 \le& \lambda^2 \frac{|\bra{\psi_\ell^{k-1}} V_k \ket{\psi_0^{k-1}}|^2}{(E_\ell^{k-1} - E_0^{k-1})^2} \\
\le& \lambda^2 \frac{|\bra{\psi_\ell^{k-1}} V_k \ket{\psi_0^{k-1}}|^2}{\gamma^2}.
\end{aligned}
\end{equation}
This allows us to write $\ket{\psi_0^{k+1}} = (1 + \delta_0) \ket{\psi_0^k} + \sum_{\ell \neq 0} \delta_\ell \ket{\psi_\ell^k}$, where $|\delta_\ell| \le \lambda \max_{k} \frac{|\bra{\psi_\ell^k} V_k \ket{\psi_0^k}|}{E_0^k - E_\ell^k}$. Letting $V_k \ket{\psi_0^k} = \kappa_k \ket{\phi_k}$, where we again choose $\kappa_k$ such that $\ket{\phi_k}$ is normalized,
\begin{equation}
\begin{aligned}
\left|\ket{\psi_0^k} - \ket{\psi_0^{k-1}}\right|^2 =& \delta_0^2 + \sum_{\ell>0} \delta_\ell^2 \le \delta_0^2 + \frac{\lambda^2}{\gamma^2} \sum_\ell |\bra{\psi_{\ell}^{k-1}} V_k \ket{\psi_0^{k-1}}|^2 \\
=& \delta_0^2 + \frac{\lambda^2}{\gamma^2} \kappa_k^2 \sum_\ell \left| \braket{\psi_\ell^k}{\phi_k} \right|^2 \\
=& \delta_0^2 + \frac{\lambda^2}{\gamma^2} \kappa_k^2.
\end{aligned}
\end{equation}
We can solve for $\delta_0^2$ in terms of $\sum_{\ell>0} \delta_\ell^2$, since $(1 + \delta_0)^2 + \sum_{\ell>0} \delta_\ell^2 = 1$. Since $\sqrt{1-x} \ge 1-x$ for $x \in [0, 1]$,
\begin{equation}
\begin{aligned}
\delta_0^2 =& \left(1 - \sqrt{1 - \sum_{\ell>0} \delta_\ell^2}\right)^2 \\
\le& \left( \sum_{\ell>0} \delta_\ell^2 \right)^2 \le \sum_{\ell>0} \delta_\ell^2
\end{aligned}
\end{equation}
since $\sum_{\ell>0} \delta_\ell^2 \le 1$.
Finally, returning to $\left|\ket{\psi_0^k} - \ket{\psi_0^{k-1}}\right|^2$, since $\kappa_k \le 1$ (this is true because $\|V_k\| = 1$), the difference between the ground states of the two Hamiltonians is at most
\begin{equation}
\left|\ket{\psi_0^k} - \ket{\psi_0^{k-1}}\right|^2 \le \frac{2\lambda^2}{\gamma^2}
\end{equation}

This means that the overlap probability between the ground states of any two adjacent Hamiltonians is $|\braket{\psi_0^{k+1}}{\psi_0^k}|^2 \ge 1 - \frac{ \lambda^2 }{ \gamma^2}$. Across $M$ segments ($M-1$ transitions), the success probability is at least $\left( 1 - \frac{ \lambda^2 }{ \gamma^2} \right)^{M - 1}$.
If we wish for the failure probability to be at most some fixed $0 < \epsilon < 1$, we must have
\begin{equation}
\begin{aligned}
&\left( 1 - \frac{ \lambda^2 }{ \gamma^2} \right)^{ M-1 } > 1 - \epsilon \\
&\frac{ \lambda }{ \gamma } < \sqrt{1 - \exp\left( \frac{\log(1-\epsilon)}{ M-1 } \right)}.
\end{aligned}
\end{equation}
\end{proof}

If we can only prepare the ground state $\ket{\psi_0}$ of the original Hamiltonian, the success probability has an additional factor $|\braket{\psi_0^1}{\psi_0}|^2$. In this case, we can apply \lem{failureprob} with $\|H - H_1\|$ included in the maximization for $\lambda$. Further, since $\gamma = \min_k (E_1^k - E_0^k) \le E_1 - E_0 - 2 \lambda | \bra{\psi_0} H_k - H\ket{\psi_0} | \le E_1 - E_0 - 2\lambda $, where $E_1 - E_0$ is the gap between the ground and first excited states of $H$, we need
\begin{equation}
\begin{aligned}
\frac{ \lambda }{ E_1 - E_0 - 2\lambda } <& \sqrt{1 - \exp\left( \frac{\log(1-\epsilon)}{ M } \right)}.
\end{aligned}
\end{equation}
Provided that this occurs, we stay in the ground state of each Hamiltonian throughout the simulation with probability $1-\epsilon$. In this case, the total accumulated phase is
\begin{equation}
\frac{\bra{\psi^{\lceil M_j \rceil}_0} e^{-i H_{\lceil M_j \rceil} \Delta t} \ket{\psi^{\lceil M_j \rceil - 1}_0} \cdots \bra{\psi^{2}_0} e^{-i H_2 \Delta t} \ket{\psi^{1}_0} \bra{\psi^{1}_0} e^{-i H_1 \Delta t} \ket{\psi_0}}{ \braket{\psi^{\lceil M_j \rceil}_0}{\psi^{\lceil M_j \rceil - 1}_0} \cdots \braket{\psi^{2}_0}{\psi^{1}_0} \braket{\psi^{1}_0}{\psi_0}} = \exp\left( -i\sum_{k=1}^{\lceil M_j \rceil} E_k^0 \Delta t \right),
\end{equation}
where $\Delta t = M_j t / \lceil M_j \rceil$.

\subsection{Phase shifts due to Hamiltonian errors}

We can generalize the analysis of the difference in the phase by determining the difference between the desired (adiabatic) unitary and the true one. Evolving under $M$ random Hamiltonians in sequence, the unitary applied for each new Hamiltonian $H_k$ is 
\begin{equation}
\begin{aligned}
U_k = \exp(-iH_k \Delta t) = \sum_\ell \ket{\psi_\ell^k} \bra{\psi_\ell^k} e^{-i E^k_\ell \Delta t},
\end{aligned}
\end{equation}
while the adiabatic unitary we would ideally apply (the unitary, assuming we remain in the same eigenstate) is
\begin{equation}
\begin{aligned}
U_{k,\text{ad}} = \sum_\ell \ket{\psi_\ell^{k+1}} \bra{\psi_\ell^k} e^{-i E^k_\ell \Delta t}.
\end{aligned}
\end{equation}
The difference between the two is that true time evolution $U_k$ under $H_k$ applies phases to the eigenstates of $H_k$, while the adiabatic unitary $U_{k,\text{ad}}$ applies the eigenphase, and then maps each eigenstate of $H_k$ to the corresponding eigenstate of $H_{k+1}$. This means that if the system begins in the ground state of $H_1$, the phase which will be applied to it by the sequence $U_{\lceil M_j \rceil,\text{ad}} U_{\lceil M_j \rceil - 1,\text{ad}} \cdots U_{2,\text{ad}} U_{1,\text{ad}}$ is proportional to the sum of the ground state energies of each Hamiltonian in that sequence. By comparison, $U_{\lceil M_j \rceil} U_{\lceil M_j \rceil - 1} \cdots U_{2} U_{1}$ will include contributions from many different eigenstates of the different Hamiltonians $H_k$.

We can bound the difference between the unitaries $U_k$ and $U_{k,\text{ad}}$ as follows.

\begin{lemma}
Let $P_0^k$ be the projector onto the ground state of $H_k$, $\ket{\psi_0^k}$, and let the assumptions of \lem{failureprob} hold. The difference between the eigenvalues of $U_k P_0^k = \exp(-i H_k\Delta t) P_0^k = \sum_\ell \ket{\psi_\ell^k} \bra{\psi_\ell^k} e^{-iE_\ell^k \Delta t} P_0^k$ and $U_{k,\text{ad}} P_0^k = \ket{\psi_0^{k+1}} \bra{\psi_0^k} e^{-i E^k_\ell \Delta t} P_0^k$, where $\Delta t$ is the simulation time, is at most
$$\| (U_k - U_{k,\text{ad}}) P_0^k \| \le \frac{2\lambda^2}{(\gamma - 2\lambda)^2}.$$
\label{lem:phaseUk}
\end{lemma}
\begin{proof}
First, we expand the true unitary using the resolution of the identity $\sum_p \ket{\psi_p^{k+1}} \bra{\psi_p^{k+1}}$, the eigenstates of the next Hamiltonian, $H_{k+1}$: 
\begin{equation}
\begin{aligned}
U_{k} = \sum_{p, \ell} \ket{\psi_p^{k+1}} \braket{\psi_p^{k+1}}{\psi_\ell^k} \bra{\psi_\ell^k} e^{-iE_\ell^k \Delta t}.
\end{aligned}
\end{equation}
Let $\Delta_{p\ell} = \braket{\psi_p^{k+1}}{\psi_\ell^k}$ for $p\neq \ell$ and $1 + \Delta_{pp} = \braket{\psi_p^{k+1}}{\psi_p^k}$ when $p=\ell$. In a sense we are writing the new eigenstate $\ket{\psi_p^{k+1}}$ as a slight shift from the state $\ket{\psi_\ell^k}$: this is the reason that we choose $\braket{\psi_p^{k+1}}{\psi_p^k} = 1 + \Delta_{pp}$. Using this definition, we can continue to simplify $U_k$, as
\begin{equation}
\begin{aligned}
U_{k} = \sum_{p} (1 + \Delta_{pp}) \ket{\psi_p^{k+1}} \bra{\psi_p^k} e^{-iE_p^k \Delta t} + \sum_{p\neq\ell} \Delta_{p\ell} \ket{\psi^{k+1}_p} \bra{\psi^k_\ell} e^{-i E_\ell^k \Delta t}.
\end{aligned}
\end{equation}
We are now well-positioned to bound $\| (U_k - U_{k,\text{ad}}) P_0^k \|$. Noting that $U_{k,\text{ad}}$ exactly equals the $1$ in the first sum in $U_k$,
\begin{equation}
\begin{aligned}
\| (U_k - U_{k,\text{ad}}) P_0^k \| =& \left\| \sum_{p} \Delta_{p\ell} \ket{\psi^{k+1}_p} \bra{\psi^k_\ell} e^{-i E_\ell^k \Delta t} \ket{\psi^k_0} \bra{\psi^k_0} \right\| \\
=& \max_{\ket\psi} \left| \sum_{p\ell} \Delta_{p\ell} \ket{\psi^{k+1}_p} \bra{\psi^k_\ell} e^{-i E_\ell^k \Delta t} \ket{\psi^k_0} \braket{\psi^k_0}\psi \right|^2 \\
=& \left| \sum_{p} \Delta_{p0} e^{-i E_0^k \Delta t} \right|^2 \\
\le& \sum_{p} |\Delta_{p0}|^2.
\end{aligned}
\end{equation}

The final step in bounding $\| U_k - U_{k,\text{ad}} \|$ is to bound $|\Delta_{p\ell}| = \braket{\psi_p^{k+1}}{\psi_\ell^k}$, similarly to how we bounded $\delta_{p\ell}$. For $p\neq\ell$, $\Delta_{p\ell}$ is given by
\begin{equation}
\begin{aligned}
|\Delta_{p\ell}|^2 = |\braket{\psi_{p}^{k+1}}{\psi_\ell^k}|^2 \le& \lambda^2 \frac{|\bra{\psi_p^k} V_k \ket{\psi_\ell^k}|^2}{(E_p^k - E_\ell^k)^2} .
\end{aligned}
\end{equation}
So, as with our bounds on $|\delta_0|^2$ and $\sum_{\ell>0} |\delta_\ell|^2$ in \lem{failureprob}, $\| (U_k - U_{k,\text{ad}}) P_0^k \|$ is upper bounded by
\begin{equation}
\begin{aligned}
\sum_p |\Delta_{p0}|^2 = \sum_p |\braket{\psi_{p}^{k+1}}{\psi_0^k}|^2 \le& \sum_p  \lambda^2 \frac{|\bra{\psi_p^k} V_k \ket{\psi_0^k}|^2}{(E_p^k - E_0^k)^2} \le \frac{2\lambda^2}{(\gamma - 2\lambda)^2},
\end{aligned}
\end{equation}
which completes the proof.
\end{proof}

We are now able to prove our main theorem, \thm{main}, which immediately follows from the prior results.

\begin{proofof}{\thm{main}}
\lem{phaseUk} gives the difference between eigenvalues of $U_k P_0^k$ and $U_{k,\text{ad}} P_0^k$. Across the entire sequence, we have
\begin{equation}
\begin{aligned}
\left\| U_{M} P_0^{M} \cdots U_k P_0^k \cdots U_1 P_0^1 - U_{M,\text{ad}} P_0^{M} \cdots U_{k,\text{ad}} P_0^k \cdots U_{1,\text{ad}} P_0^1 \right\| \le \frac{2 M \lambda^2}{(\gamma - 2\lambda)^2}.
\end{aligned}
\end{equation}
This is the maximum possible difference between the accumulated phases for the ideal and actual sequences, assuming the system leaves the ground state for at most one repetition at a time.

The probability of leaving the ground state as part of a Landau-Zener process instigated by the measurement at adjacent values of the Hamiltonians is, under the assumptions of~\lem{failureprob}, that the failure probability occuring at each projection is $\epsilon$  if

\begin{equation}
 \frac{ \lambda }{ \gamma } < \sqrt{1 - \exp\left( \frac{\log(1-\epsilon)}{ M-1 } \right)},
\end{equation}
thus the result follows trivially from these two results.
\end{proofof}

\section{Importance sampling}
\label{app:importancesampling}
Here we prove our main lemma regarding importance sampling.
\begin{proofof}{\lem{robust}}
The proof is a straightforward exercise in the triangle inequality once one uses the fact that $|\delta_j| \le |F(j)|/2$ and the fact that $1/(1-|x|)\le 1 +2|x|$ for all $x\in [-1/2,1/2]$.

First note that because importance sampling introduces no bias, $\mathbb{E}_f(F) = \mathbb{E}(F)$, and
\begin{equation}
\mathbb{V}_f(F) = \mathbb{E}\left(\left(\frac{F}{N f}\right)^2 \right) - \mathbb{E}(F)^2 = \sum_j f(j) \left(\frac{F(j)}{N f(j)}\right)^2 -\mathbb{E}(F)^2.
\end{equation}
Next recall that $f = |\widetilde{F}|/\sum_k |\widetilde{F}(k)|$ with $|\widetilde{F}(j)| = |F(j)| +\delta_j$.  From this we find that 
\begin{align}
\mathbb{V}_f(F) &= \frac{1}{N^2}\left(\sum_k|F(k)| + \delta_k \right)\left(\sum_j \frac{F^2(j)}{|F(j)| + \delta_j}\right) - \left(\mathbb{E}(F)\right)^2\nonumber\\
&\le\frac{1}{N^2}\left(\sum_k|F(k)| + \delta_k \right)\left(\sum_j \frac{F^2(j)}{|F(j)| - |\delta_j|}\right) - \left(\mathbb{E}(F)\right)^2\nonumber\\
&\le \frac{1}{N^2}\left(\sum_k|F(k)| + |\delta_k| \right)\left(\sum_j |F(j)| + 2|\delta_j|\right) - \left(\mathbb{E}(F)\right)^2\nonumber\\
&= \frac{1}{N^2}\left(\sum_k|\delta_k| \right)\left(\sum_j |F(j)| + 2|\delta_j|\right)+ \frac{1}{N^2}\left(\sum_k|F(k)| \right)\left(2\sum_j|\delta_j|\right) +\left(\mathbb{E}(|F|)\right)^2- \left(\mathbb{E}(F)\right)^2\nonumber\\
&\le\frac{4}{N^2}\left(\sum_k|\delta_k| \right)\left(\sum_j |F(j)|\right)+\mathbb{V}_{f_{\rm opt}}(F).
\end{align}
\end{proofof}

This bound is tight in the sense that as $\max_k |\delta_k|\rightarrow 0$ the upper bound on the variance converges to $\left(\mathbb{E}(|F|)\right)^2- \left(\mathbb{E}(F)\right)^2$, which is the optimal attainable variance.

\end{document}